\documentclass[runningheads]{llncs}
\usepackage{fullpage}
\usepackage{graphicx}
\usepackage{comment}
\usepackage{epsfig,epsf}
\usepackage{epstopdf}
\usepackage{graphics}
\usepackage{array}
\usepackage{amsmath}
\usepackage{amssymb}
\usepackage{comment}

\usepackage{fullpage}
\usepackage{multirow}
\usepackage{enumitem}
\usepackage{algcompatible}
\usepackage[T1]{fontenc}
\usepackage{bm}
\usepackage{xcolor}
\usepackage[ruled,vlined,linesnumbered]{algorithm2e}
\date{}

\begin{document}
\title{Hamiltonicity: Variants and Generalization in $P_5$-free Chordal Bipartite graphs\thanks{This work is partially supported by DAE-NBHM project, NBHM/ 02011/37/2017/RD II/16646}}
%
%
\author{S.Aadhavan \inst{1} \and
	R.Mahendra Kumar \inst{1} \and
	P.Renjith \inst{2} \and
	N.Sadagopan \inst{1}}
\authorrunning{N.Sadagopan  et al.}
%
\institute{Indian Institute of Information Technology Design and Manufacturing, Kancheepuram, Chennai, India \\\email{aadhav1395@gmail.com,\{coe18d004,sadagopan\}@iiitdm.ac.in} \and Indian Institute of Information Technology Design and Manufacturing, Kurnool. \\
	\email{renjith@iiitk.ac.in}}
\maketitle              
\begin{abstract}
A bipartite graph is chordal bipartite if every cycle of length at least six has a chord in it. M$\ddot{\rm u}$ller \cite {muller1996Hamiltonian}  has shown that the Hamiltonian cycle problem is NP-complete on chordal bipartite graphs by presenting a polynomial-time reduction from the satisfiability problem.  The microscopic view of the reduction instances reveals that the instances are $P_9$-free chordal bipartite graphs, and hence the status of Hamiltonicity in $P_8$-free chordal bipartite graphs is open.  In this paper, we identify the first non-trivial subclass of $P_8$-free chordal bipartite graphs which is $P_5$-free chordal bipartite graphs, and present structural and algorithmic results on $P_5$-free chordal bipartite graphs.  We investigate the structure of $P_5$-free chordal bipartite graphs and show that these graphs have a {\em Nested Neighborhood Ordering (NNO)}, a special ordering among its vertices.  Further, using this ordering, we present polynomial-time algorithms for classical problems such as the Hamiltonian cycle (path), also the variants and generalizations of the Hamiltonian cycle (path) problem. We also obtain polynomial-time algorithms for treewidth (pathwidth), and minimum fill-in in $P_5$-free chordal bipartite graph.  We also present some results on complement graphs of $P_5$-free chordal bipartite graphs.

\keywords{$P_5$-free chordal bipartite graphs \and Nested Neighborhood Ordering \and Hamiltonian cycle (path) \and Longest cycle (path) \and Steiner cycle (path) \and Treewidth \and Pathwidth \and Minimum fill-in}
\end{abstract}
\section{Introduction}
The graphs with forbidden subgraphs possess a nice structural characterization.  The structural characterization of these graphs has attracted researchers from both mathematics and computing. The popular graphs are chordal graphs \cite{dirac1961rigid}, which forbid induced cycles of length at least four, and chordal bipartite graphs \cite{fulkerson1965incidence}, which are bipartite graphs that forbid induced cycles of length at least six.  Further, these graph classes have a nice structural characterization with respect to minimal vertex separators and a special ordering, namely, perfect elimination ordering \cite{golumbie1980algorithmic} among its vertices (edges).  These graphs are largely studied in the literature to understand the computational complexity of classical optimization problems such as vertex cover, dominating set, coloring, etc., as these problems are known to be NP-complete on general graphs.  Thus these graphs help to identify the gap between NP-complete instances and polynomial-time solvable instances for classical combinatorial problems.\\\\
The Hamiltonian cycle (path) problem is a famous problem that asks for the presence of a cycle (path) that visits each node exactly once in a graph. A graph $G$ is said to be Hamiltonian if it has a Hamiltonian cycle. The Hamiltonian problem plays a significant role in various research areas such as circuit design \cite{wang2012efficient}, operational research \cite{malakis1976Hamiltonian}, biology \cite{dorninger1994Hamiltonian}, etc.  On the complexity front, this problem is well studied and it remains NP-complete on general graphs.  Interestingly, this problem is polynomial-time solvable on special graphs such as cographs and permutation graphs \cite{deogun1994polynomial}.  Surprisingly, this problem is NP-complete on chordal graphs \cite{bertossi1986Hamiltonian}, bipartite graphs \cite{Krishnamoorthy:1975:NPB:990518.990521}, and $P_5$-free graphs \cite{golumbie1980algorithmic}. \\ \\
M$\ddot{\rm u}$ller \cite{muller1996Hamiltonian}  has shown that the Hamiltonian cycle problem is NP-complete on chordal bipartite graphs by presenting a polynomial-time reduction from the satisfiability problem.  The microscopic view of the reduction instances reveals that the instances are $P_9$-free chordal bipartite graphs.  It is natural to study the complexity of the Hamiltonian cycle problem in $P_8$-free chordal bipartite graphs and its subclasses.  Since $P_4$-free chordal bipartite graphs are complete bipartite graphs, the first non-trivial graph class in this line of research is $P_5$-free chordal bipartite graphs.  Further,  the Steiner tree problem and the dominating set problem are also NP-complete for chordal bipartite graphs \cite{mullar1987brandstadt} as there is a polynomial-time reduction from the vertex cover problem.  Interestingly, the instances generated have $P_5$ in it.  This gives us another motivation to study the structure of $P_5$-free chordal bipartite graphs.  \\ \\
It is known from the literature that the Hamiltonian cycle problem is NP-complete on $(C_4,C_5,P_5)$-free graphs \cite{golumbie1980algorithmic}. Therefore, it is NP-complete on $(C_5,P_5)$-free graphs as well. In this paper, we present a polynomial-time algorithm for $(C_3,C_5,P_5)$-free graphs which are precisely $P_5$-free chordal bipartite graphs. \\ \\
A graph $G$ is $k$-bisplit if it has a stable set (an independent set) $X$ such that $G-X$ has at most $k$ bicliques (complete bipartite graphs). A graph $G$ is weak bisplit if it is $k$-bisplit for some $k$. It is known from \cite{bisplit}, recognizing weak bisplit graphs is NP-complete, whereas recognition of $k$-bisplit graphs for every fixed $k$ is polynomial-time solvable. A graph $G$ is bisplit if its vertex set can be partitioned into three stable sets $X$, $Y$, and $Z$ such that  $Y\cup Z$ induces a complete bipartite subgraph (a biclique) in $G$. Note that 1-bisplit graphs are bisplit graphs.  Brandst{\"a}dt et al. \cite{bisplit} provided $\mathit{O(nm)}$ linear time recognition algorithm for bisplit graphs. Subsequently, Abueida et al. \cite{bisplitn2} has given $\mathit{O(n^2)}$ time recognition algorithm for bisplit graphs.  Various popular problems like, Hamiltonian cycle (path) problem, domination, independent dominating set and feedback vertex set are NP-complete on bisplit graphs. So, it is natural to investigate the complexities of the problems mentioned above in a subclass of bisplit graphs. Interestingly, chordal bipartite graphs are a subclass of bisplit graphs and hence due to \cite{muller1996Hamiltonian} the Hamiltonian cycle (path) problem is NP-complete on this graph class.  This is another motivation to investigate the structure of $P_5$-free chordal bipartite graphs, which are a subclass of bisplit graphs.\\\\
In recent times, the class of $P_5$-free graphs have received a good attention, and  problems such as independent set \cite{Lokshantov:2014:ISP:2634074.2634117}, k-colourbility \cite{hoang2010deciding}, and feedback vertex set \cite{fvsp5} have polynomial-time algorithms restricted to $P_5$-free graphs. \\ \\
The longest path problem is the problem of finding an induced path of maximum length in $G$. Since this problem is a generalization of the Hamiltonian path problem, it is NP-complete on general graphs. This problem is polynomial-time solvable in interval graphs \cite{lpathinterval} and biconvex graphs \cite{lpathbiconvex}. A minimum leaf spanning tree is a spanning tree of $G$ with the minimum number of leaves.  This problem is NP-complete on general graphs, and polynomial-time solvable on cographs \cite{cographsmlst}. It is important to highlight that the minimum leaf spanning tree problem in cographs can be solved using the longest path problem as a black box. In this paper, we wish to see whether we can come up with a framework for Hamiltonicity and its variants in $P_5$-free chordal bipartite graphs.    \\\\
It is important to highlight that other classical problem such as Steiner tree, dominating set, connected dominating set are known to be NP-complete on chordal bipartite graphs \cite{mullar1987brandstadt}.\\\\
Similar to the Hamiltonian cycle (path) problem, its variants also attracted many researchers. We shall see some of the popular variants of the Hamiltonian cycle (path) problem.  A graph $G$ with the vertex set $V(G)$ and the edge set $E(G)$ is pancyclic \cite{bondy1971pancyclic} if it contains cycles of all lengths $l$, $3 \leq l \leq |V(G)|$. A graph $G$ is (2)-pancyclic graph \cite{2pancycle} if it contains exactly two cycles of length $l$, $3 \leq l \leq |V(G)|$. A graph $G$ is said to be bipancyclic \cite{bipancyclic} if it contains a cycle of length $2l$, $2 \leq l \leq |V(G)|$. A graph $G$ is said to be Hamiltonian connected if there exists a Hamiltonian path between every pair of vertices. Asratian \cite{asratian1993criterionhamiltonconnected} and Kewenet et al. \cite{kewen2008hamiltonianhamiltonconnected} have given sufficient conditions for a graph to be Hamiltonian connected. A graph $G$ is said to be homogeneously traceable if there exists a Hamiltonian path beginning at every vertex of $G$.  Skupie{\'n} \cite{skupien1984homogeneously} has given necessary condition for a graph $G$ to be homogeneously traceable. Chartrand et al.\cite{chartrand1979homogeneously} proved that there exists a homogeneously traceable non-Hamiltonian graph of order $n$ for all positive integers $n$ except $3 \leq n \leq 8$. A hypo-Hamiltonian graph is a non-Hamiltonian graph $G$ such that $G - v$ is Hamiltonian for every vertex $v \in V (G)$.  A graph is $k$-ordered Hamiltonian if for every ordered sequence of $k$ vertices, there exists a Hamiltonian cycle that encounters the vertices of the sequence in the given order. Faudree et al. \cite{korderhamilton} have given degree conditions for a graph to be $k$-ordered Hamiltonian. These variants are known to be NP-complete on general graphs, and chordal biparatite graphs.   \\\\
In the literature, there are many conditions for Hamiltonicity and its variants. Due to the inherent difficulty of the problem, all these conditions are either necessary conditions or sufficient conditions, however no known necessary and sufficient conditions. 
We show that Chv{\'a}tal's Hamiltonian cycle (path) necessary condition is sufficient for $P_5$-free chordal bipartite graphs. We also present a necessary and sufficient condition for the existence of the Steiner path (cycle) in $P_5$-free chordal bipartite graphs.\\\\
{\bf Our work:} In this paper, we study the structure of $P_5$-free chordal bipartite graphs and present a new ordering referred to as {\em Nested Neighbourhood Ordering} among its vertices. We present a polynomial-time algorithm for the Hamiltonian cycle (path) problem using the {\em Nested Neighbourhood Ordering}.  Further, using this ordering, we present polynomial-time algorithms for longest path (cycle),  minimum leaf spanning tree,  Steiner path (cycle), treewidth (pathwidth), and minimum fill-in.  We believe that these results can help us in the study of other combinatorial problems restricted to $P_5$-free chordal bipartite graphs, and also in the structural investigation of $P_k$-free chordal bipartite graphs, $6 \leq k \leq 8$.   Further, our understanding shall help us in obtaining a dichotomy: easy vs hard input instances of chordal bipartite graphs. We also present some results on complement graphs of $P_5$-free chordal bipartite graphs. \\ \\ 
{\bf Related work:} Structure of $P_5$-free graphs has been looked at while studying the Difference graph \cite{diffgraph1990}, 3-colorable $P_5$-free graphs \cite{maffray3color2012}, and $(P_5, \overline{P_5} )$-free graphs \cite{fouquet1993decomposition}.  However, the explicit mention of  {\em NNO} has not been reported in the literature.  To the best of our knowledge, the algorithmic results presented in this paper are new and have not been reported in the literature.
\subsection{Preliminaries}
In this paper, we work with simple, connected, undirected and unweighted graphs.   For a graph $G$, let $V(G)$ denote the vertex set and $E(G)$ denote the edge set.  The notation $\{u,v\}$ represents an edge incident on the vertices $u$ and $v$.  The neighborhood of a vertex $v$ of $G$, $N_G(v)$, is the set of vertices adjacent to $v$ in $G$.  The degree of a vertex $u$ is denoted by $d_G(u)=|N_G(u)|$. We use $d(u)$ and $d_G(u)$ interchangeably if the graph under consideration is clear from the context. A graph $H$ is called an induced subgraph of $G$ if for all
$u, v \in V (H)$, $\{u, v\} \in E(H)$ if and only if $\{u, v\} \in E(G)$.  The graph induced on $V(G)\setminus\{u\}$ is denoted by $G-u$. A graph $G$ is said to be connected if there exists a path between every pair of vertices. If $G$ is disconnected, then $c(G)$ denotes the number of connected components in $G$ (each component being maximal).  A bipartite graph is chordal bipartite if every cycle of length at least six has a chord.  A maximal biclique $K_{i,j}$ is a complete bipartite graph such that there is no strict supergraphs $K_{i+1,j}$ or $K_{i,j+1}$.  A maximum biclique is a maximal biclique with the property that $|i-j|$ is minimum.  Note that $P_5$ is an induced path of length $5$ and $P_{uv}$ denotes a path that starts at $u$ and ends at $v$, and $|P_{uv}|$ denotes its length.  We use $P_{uv}$ and $P(u,v)$ interchangeably. 
Let $\overline{G}$ denote the complement of the graph $G$, where $V (\overline{G}) = V (G)$ and $E(\overline{G}) = \{\{u, v\} | \{u, v\} \notin E(G)\}$.
\section{Structural Results}
In this section, we shall present a structural characterization of $P_5$-free chordal bipartite graphs.  Also, we introduce {\em Nested Neighborhood Ordering (NNO)} among its vertices.  We shall fix the following notation to present our results.  For a chordal bipartite $G$ with bipartition $(A,B)$, let $A=\{x_1,x_2,\ldots,x_m\}$ and $B=\{y_1,y_2,\ldots,y_n\}$. The edge set $E(G) \subseteq \{\{x,y\} ~|~ x \in A,y \in B \}$. Let $A=A_{1} \cup A_{2}$ and $B=B_{1} \cup B_{2}$, $A_{1}=\{x_{1},x_{2},\ldots,x_{i}\}$, $A_{2}=\{x_{i+1},x_{i+2},\ldots,x_{m}\}$ and $B_{1}=\{y_1,y_2,\ldots,y_j\}$, $B_{2}=\{y_{j+1},y_{j+2},\ldots,y_{n}\}$, such that $(A_1,B_1)$ is a maximum biclique.  
\begin{lemma}
	\label{lem1}
	Let $G$ be a $P_5$-free chordal bipartite graph.  Then, $\forall x \in A_{2},\exists y_{k} \in B_{1}$ such that $y_{k} \not\in N_G(x)$ and $\forall y \in B_{2},\exists x_{k} \in A_{1}$ such that $x_{k} \not\in N_G(y)$.
\end{lemma}
\begin{proof}
	Suppose there exists $x$ in $A_2$ such that for all $y_k$ in $B_1$, $xy_k \in E(G)$.  Then $(A_1 \cup \{x\},B_1)$ is the maximum biclique, contradicting the fact that $(A_1,B_1)$ is maximum.  Similar argument is true for $y \in B_2$.  Therefore, the lemma follows. \qed 
\end{proof}
\begin{lemma}
	\label{lem2}
	Let $G$ be a $P_5$-free chordal bipartite graph.  Then,  $\forall x \in A_{2}$, $N_G(x) \subset B_{1}$ and $ \forall y \in B_{2}$, $N_G(y) \subset A_{1}$. 
\end{lemma}
\begin{proof}
	On the contrary,  $\exists x_{a} \in A_{2}$, $N(x_{a}) \not\subset B_{1}$. Case 1: $N(x_{a}){=}B_{1}$. Then, $(A_{1} \cup \{x_{a}\}, B_{1})$ is the  maximum biclique, a contradiction.  Case 2: $N(x_{a}) \subseteq B_{2}$.  This implies that there exists $y_{b} \in B_2$ such that $y_b \in N(x_{a})$.  Since $G$ is connected, $N(y_{b}) \cap A_{1} \neq \emptyset$, say $x_{c} \in N(y_{b}) \cap A_{1}$.  In $G$, $P(x_{a},x_{k})=(x_{a},y_{b},x_{c},y_{k},x_{k})$ is an induced $P_{5}$.  Note that, due to the maximality of $(A_1,B_1)$, as per Lemma \ref{lem1}, we find $x_{k} \not\in N(y_{b}), y_{k} \not\in N(x_{a})$.  This  contradicts that $G$ is $P_5$-free.  Case 3: $N(x_{a}) \cap B_{1} \neq \emptyset$ and $N(x_{a}) \cap B_{2} \neq \emptyset$. I.e., $ \exists x_{a} \in A_{2}$ such that $y_{b},y_{c} \in N(x_{a})$ and $y_{b} \in B_{1}$, $y_{c} \in B_{2}$.  In $G$, $P(y_{c},y_{k})=(y_{c},x_{a},y_{b},x_{k},y_{k}$), $x_{k} \not\in N(y_{c})$, $y_{k} \not\in N(x_{a})$ is an induced $P_{5}$.  Note that the existence of $x_k,y_k$ is due to Lemma \ref{lem1}.  This is contradicting the $P_5$-freeness of $G$.  Similarly, $ \forall y \in B_{2}$, $N(y) \subset A_{1}$, can be proved. \qed
\end{proof}
\begin{lemma}
	\label{lem3}
	Let $G$ be a $P_5$-free chordal bipartite graph.  
	For $x_i,x_j \in A_{2}$, $i \not= j$, if $d(x_i) \leq d(x_j)$, then $N(x_i) \subseteq N(x_j)$.   Similarly, for $y_i,y_j \in B_{2}$, if $d(y_i) \leq d(y_j)$, $N(y_i) \subseteq N(y_j)$.
\end{lemma}
\begin{proof}
	Let us assume to the contrary that $N(x_i) \not\subseteq N(x_j)$.  I.e., $N(x_i)  \setminus  N(x_j) \neq \emptyset$. \newline \textbf{Case 1:} $N(x_i) \cap N(x_j)\neq \emptyset$. I.e., $\exists y_{a}\in B$ such that  $y_{a} \in N(x_i) \cap N(x_j)$.  Since $N(x_i) \not\subseteq N(x_j)$, $\exists y_{b}\in B$ such that  $y_{b} \not\in N(x_i) \cap N(x_j)$ and $y_{b} \in N(x_i)$.  Since $d(x_j) \geq d(x_i)$, vertex $x_j$ is adjacent to at least one more vertex $y_{c} \in B$ such that $y_{c} \not\in N(x_i)$.  The path $P(y_{c},y_{b})=(y_{c},x_j,y_{a},x_i,y_{b}$) is an induced $P_{5}$.  This is a contradiction. \newline \textbf{Case 2:} $N(x_i) \cap N(x_j)= \emptyset $, $ | N(x_i) | \geq 1$ and $ | N(x_j) | \geq 1$. I.e., $\exists y_{a},y_{b}\in B$ such that  $y_{a} \in N(x_i)$, $y_{a} \not\in N(x_j)$ and $y_{b} \in N(x_j)$, $y_{b} \not\in N(x_i)$.  Since $G$ is a connected graph,  $|P(y_{a},y_{b})| \geq 3$, an induced path of length at least 3.  The path $P(x_i,x_j)=(x_i,P(y_{a},y_{b}),x_j$) has an  induced  path of length at least 5. This is a contradiction.  Similarly, for all pairs of distinct vertices $y_i,y_j \in B_{2}$ with $d(y_i) \leq d(y_j)$, $N(y_i) \subseteq N(y_j)$ can be proved.   \qed
\end{proof}
\begin{theorem}
	\label{thm1}
	Let $G$ be a $P_5$-free chordal bipartite graphs with $(A_1,B_1)$ being the maximum biclique.  Let $A_{2}=(u_{1},u_{2},...,u_{p})$ and $B_{2}=(v_{1},v_{2},...,v_{q})$ are orderings of vertices.  If $d_{G}(u_1) \leq d_{G}(u_2) \leq d_{G}(u_3) \leq \ldots \leq d_{G}(u_p)$, then $N(u_{1}) \subseteq N(u_{2}) \subseteq N(u_{3}) \subseteq \ldots \subseteq N(u_{p})$.  Further, if $d_{G}(v_1) \leq d_{G}(v_2) \leq d_{G}(v_3) \leq \ldots \leq d_{G}(v_q)$, then $N(v_{1}) \subseteq N(v_{2}) \subseteq N(v_{3}) \subseteq \ldots \subseteq N(v_{q})$.
\end{theorem} 
\begin{proof} 
	We shall prove by mathematical induction on $|A_2|$.  Base Case: $|A_{2}|=2, A_2=(u_{1},u_{2})$ such that $d(u_{1}) \leq d(u_2)$.  By Lemma \ref{lem3}, $N(u_1) \subseteq N(u_2)$.  Induction step: Consider $A_2=(u_{1},u_{2},u_{3},\ldots,u_{p}), p \geq 3$.  Consider the vertex $u_{p} \in A_{2}$ such that $d(u_p) \geq d(u_{p-1})$.  By Lemma \ref{lem3}, $N(u_{p-1}) \subseteq N(u_{p})$ is true.  By the hypothesis, $N(u_{1}) \subseteq N(u_{2}) \subseteq N(u_{3}) \subseteq \ldots \subseteq N(u_{p-1})$.  By combining the hypothesis and the fact that $N(u_{p-1}) \subseteq N(u_{p})$, our claim follows.  Similarly for $B_{2}$ as well.\qed
\end{proof}
We refer to the above ordering of vertices as {\em Nested Neighbourhood Ordering (NNO)} of $G$.  From now on, we shall arrange the vertices in $A_2$ in non-decreasing order of their degrees so that we can work with {\em NNO} of $G$.
\begin{lemma}
	\label{bisplittop5}
	Let $G$ be a $P_5$-free chordal bipartite graph. Then, $G$ is bisplit.
	\end{lemma}
\begin{proof}
	By our structural result, we know that $G$ is partitioned into four stable sets namely $A_1$, $B_1$, $A_2$, and $B_2$. We observe that $A_1 \cup B_1$ induces a maximum biclique and $A_2$, $B_2$ is an independent set. Therefore, by definition, $G$ is bisplit. \qed
	\end{proof}
\section{Hamiltonicity in $P_{5}$-free Chordal Bipartite graphs}
In this section, we shall present polynomial-time algorithms for the Hamiltonian cycle (path) problem in $P_5$-free chordal bipartite graphs.  For a connected graph $G$ and set $S \subset V(G)$, $c(G-S)$ denotes the number of connected components in the graph induced on the set $V(G) \setminus S$.  It is well-known, due to, Chvatal \cite{dbwest2003} that if a graph $G$ has a Hamiltonian cycle, then for every $S\subset V(G), c(G-S) \leq |S|$. Similarly, if a graph $G$ has a Hamiltonian path, then for every $S\subset V(G), c(G-S) \leq | S|+1$. 
\begin{theorem}
	\label{thm2}
	For a $P_5$-free chordal bipartite graph $G$, $G$ has a  Hamiltonian cycle if and only if (i) $|A| = |B|$ and (ii) $A_{2}$ has an ordering $(u_{1},u_{2},\ldots,u_{p})$, such that $\forall u_g, d(u_g) > g$, $1 \leq g \leq p$ and $B_{2}$ has an ordering $(v_{1},v_{2},\ldots,v_{q})$, $ \forall v_h, d(v_h) >h$, $1 \leq h \leq q$.
\end{theorem}
\begin{proof}
	{\em Necessity:} (i) Any cycle in a bipartite graph is even and alternates between vertices of $A$ and $B$. Since the Hamilton cycle visits all the vertices in $A$ and $B$, it must have $ |A|= |B|$. (ii) On the contrary, $ \exists u_g \in A_{2}$  such that $u_g$ is the first vertex in the ordering with $d(u_g) \leq g$.  That is, for $u_k \in \{u_1,\ldots,u_{g-1}\}$, $d_{G}(u_k) > k$ and $d(u_g) \leq g$.  Since $G$ follows {\em NNO}, $d_{G}(u_g)= g$.  From Theorem \ref{thm1},  we know that $N(u_1) \subseteq N(u_2) \subseteq \ldots \subseteq N(u_{g-1}) \subseteq N(u_g)$.   This implies that $c(G-N(u_g)) = g+1 >g$.  This is a contradiction to Chvatal's necessary condition for the Hamiltonian cycle. Similarly, in $B_{2}$, $ \forall v_h,$ $d(v_h)>h$ can be proved. \\
	{\em Sufficiency:} Let $i=|A_1|$ and $j=|B_1|$. Since $A_{2}$ has an ordering such that $ \forall u_g \in A_2$, $d(u_g)>g$, for clarity purpose, we define $N_G(u_g)$ as follows; $N(u_g)=\{y_{1},y_{2},\ldots,y_{l}\}$, $g<l<j$,  that is, $u_1$ is adjacent to at least two vertices $\{y_{1},y_{2}\}$ and at most $j-1$ vertices $\{y_{1},y_{2},\ldots,y_{j-1}\}$, $u_2$ is adjacent to at least three vertices $\{y_{1},y_{2},y_{3}\}$ and at most $j-1$ vertices $\{y_{1},y_{2},\ldots,y_{j-1}\}$ and similarly $u_p$ is adjacent to at least $p+1$ vertices $\{y_{1},y_{2},\ldots,y_{p+1}\}$ and at most $j-1$ vertices $\{y_{1},y_{2},\ldots,y_{j-1}\}$.  Observe that, due to the maximality of $(A_1,B_1)$, any $u_g$ of $A_2$ can be adjacent to at most $j-1$ vertices of $B_1$.  Similarly, in $B_2$, for all $v_h \in B_2$,               $N(v_h)=\{x_{1},x_{2},\ldots,x_{l}\},$ $h<l<i$. \\
	Let $d(u_p)=r, p+1 \leq r \leq j-1$ and $d(v_q)=s, q+1 \leq s \leq i-1$. The vertices in $A_{1}$ can be ordered as $(x_{1},x_{2},\ldots,x_{q},x_{q+1},\ldots,x_{i})$ and the vertices in $B_{1}$ can be ordered as $(y_{1},y_{2},\ldots,y_{p},y_{p+1},\ldots,y_{j})$.  
	Note that $A_{3}=A_{1}{\setminus}\{x_{1},x_{2},\ldots,x_{q+1}\}=\{x_{q+2},\ldots,x_{i-1},x_{i}\}$ and $B_{3}=B_{1}{\setminus}\{y_{1},y_{2},\ldots,y_{p+1}\}=\{y_{p+2},\ldots,y_{j-1},y_{j}\}$. 
	Further, $|A_{3}| = | A | -(| A_{2} | +q+1)= | A | -(p+q+1)$ and $ | B_{3} | = | B | -( | B_{2} | +p+1)= | B | -(q+p+1)$. 
	Since $ | A | = | B | $, it follows that $ | A_{3} | = | B_{3} | $. 
	In $G$,  $(y_{1},u_{1},y_{2},u_{2},\ldots,y_{p},u_{p},y_{p+1},x_{1},v_{1},x_{2},v_{2},\ldots,x_{q},v_{q},x_{q+1},y_{p+2},x_{q+2},$ $\ldots,y_{j},x_{i},y_{1})$ is a Hamiltonian cycle. \qed
\end{proof}
\begin{theorem}
	\label{thm3}
	For a $P_5$-free chordal bipartite graph $G$, $G$ has a  Hamiltonian path  if and only if one of the following is true \\
	(i)  $|A|=|B|$ and $A_2$ has an ordering,  $\forall u_g, d(u_g) \geq g$, $1 \leq g \leq p$ and $B_{2}$ has an ordering, $ \forall v_h, d(v_h) \geq h$, $1 \leq h \leq q$. \\
	(ii) $|A|=|B|+1$ and $A_2$ has an ordering,$\forall u_g, d(u_g) \geq g$, $1 \leq g \leq p$ and $B_{2}$ has an ordering, $ \forall v_h, d(v_h) >h$, $1 \leq h \leq q$. 
\end{theorem}
\begin{proof}
	{\em Necessity:} (i) Without loss of generality, we assume $A \geq B$.  Any Hamiltonian path starting at $A$ and alternates between $A$ and $B$ can end at $A$ or $B$.  Therefore $|A|=|B|$ or $|A|=|B|+1$. To prove that $A_2$ satisfies the ordering, we assume to the contrary that $ \exists u_g \in A_{2}$ such that $u_g$ is the first vertex in the ordering such that $d(u_g)<g$.  Since $G$ follows {\em NNO}, $d_{G}(u_g)= g - 1$. From Theorem \ref{thm1}, we know that $N(u_1) \subseteq N(u_2) \subseteq \ldots \subseteq N(u_{g-1}) \subseteq N(u_g)$.  Note that, as per the ordering of $A_2$, $N(u_{g-1})=N(u_g)$. On removing $N(u_g)$ from $G$ we have $g$ components in $A_2$ and $A_1 \cup B_2 \cup (B_1 -N(u_g))$ forms another component.  This implies that $c(G-N(u_g)) =  g+1$.  Clearly, $g+1 \not \leq g-1$.  Thus we contradict the Chvatal's necessary condition for the Hamiltonian path.
	Similarly $B_{2}$ has an ordering such that $ \forall v_h, d(v_h) \geq h$, $1 \leq h \leq q$. \\
	(ii) For $A_2$, the argument is similar to the above.   Suppose ${\exists}v_{r}{\in}B_{2}$ such that $v_{r}$ is the first vertex in the ordering such that $d(v_{r}){\leq}r$. From Theorem \ref{thm1}, $N(v_{1}){\subseteq}N(v_{2}){\subseteq}\ldots{\subseteq}N(v_{r-1}){\subseteq}N(v_{r})$. Consider the set $S=B_{1}{\cup}\{v_{r+1},v_{r+2},\ldots,v_{q}\}$ and $|S|=j+q-r-1+1=j+q-r$.  Further, $c(G-S) {\geq} p+1+i-r$.  Note that $|A|=i+p$ and $|B|=j+q$.  Since $|A|=|B|+1$, $c(G-S) {\geq} p+1+i-r = j+q+1+1-r=j+q-r+2$.  Clearly, $c(G-S) \not \leq |S|+1$, contradicting the Chvatal's condition for the Hamiltonian path. \\
	{\em Sufficiency:} (i) Let $N(u_g)=\{y_{1},\ldots,y_{l}\}$, $g \leq l < j$ and $N(v_h)=\{x_{1},\ldots,x_{l}\},$ $h{\leq}l<i$. Consider $A_{3}=A_{1}{\setminus}\{x_{1},x_{2},\ldots,x_{q}\}=\{x_{q+1},x_{q+2},\ldots,x_{i-1},x_{i}\}$.  $B_{3}=B_{1}{\setminus}\{y_{1},y_{2},\ldots,y_{p}\}=\{y_{p+1},x_{p+2},\ldots,y_{j-1},y_{j}\}$.   Note that $|A_3|=|A|-(p+q)$ and $|B_3|=|B|-(p+q)$.  In $G$, \\$P(u_{1},y_{1},u_{2},y_{2},\ldots,u_{p},y_{p},x_{q+1},y_{p+1},x_{q+2},y_{p+2},\ldots,x_{i},y_{j},x_{q},v_{q},\ldots,x_{1},v_{1})$ is a Hamiltonian path.  \\
	(ii) Consider $A_{3}=A_{1}{\setminus}\{x_{1},x_{2},\ldots,x_{q+1}\}=\{x_{q+2},x_{q+3},\ldots,x_{i-1},x_{i}\}$ and 
	$B_{3}=B_{1}{\setminus}\{y_{1},y_{2},\ldots,y_{p}\}=\{y_{p+1},x_{p+2},\ldots,y_{j-1},y_{j}\}$.
	In $G$, $P(u_{1},y_{1},u_{2},y_{2},\ldots,u_{p},y_{p},x_{q+2},y_{p+1},x_{q+3},y_{p+2},\ldots,x_{i},y_{j},x_{q+1},v_{q},\ldots,x_{2},$ $v_{1},x_{1})$ is a Hamiltonian path.  This completes the proof of this claim. \qed
\end{proof}
\begin{theorem}
	Let $G$ be a $P_5$-free chordal bipartite graph.  Finding a Hamiltonian path and cycle in $G$ are polynomial-time solvable.
\end{theorem}
\begin{proof}
	Follows from  the characterizations presented in Theorems \ref{thm2} and \ref{thm3} as the proofs are constructive. \qed
\end{proof}
\section{Chv{\'a}tal's Necessary condition is Sufficient on $P_5$-free chordal bipartite graphs}
For the Hamiltonian cycle (path) problem, it is well-known from the literature that there are no necessary and sufficient conditions. However, we know the condition due to Chv{\'a}tal is necessary but not sufficient and results due to Ore and Dirac are sufficient but not necessary. In this section, we show that Chv{\'a}tal's necessary condition is sufficient for $P_5$-free chordal bipartite graphs.

\begin{theorem}
	Let $G$ be a $P_5$-free chordal bipartite graph with $|A| = |B|$. If $G$ satisfies $c(G-S)\leq|S|$, for every non-empty subset $S$ $\subseteq V(G)$, then $G$ has the Hamiltonian cycle.  
\end{theorem}
\begin{proof}
	Assume on the contrary, that $G$ has no Hamiltonian cycle, then there exists $u_g$ or $v_h$ in $A_2 (B_2)$ that violates the degree conditions mentioned in Theorem 2.  Let $u_g$ be the first vertex in the ordering with $d(u_g)\leq g$.  By Theorem \ref{thm1}, we know that $G$ follows {\em NNO}, $d_{G}(u_g)= g$.  This implies that $c(G-N(u_g)) = g+1 >g$, which is a contradiction to the premise of the theorem. Similarly, $v_h$ in $B_2$ can be proved. \qed
	
\end{proof}

\begin{theorem}
	Let $G$ be a $P_5$-free chordal bipartite graph with $|A| = |B|$ or $|A| = |B|+1$ . If $G$ satisfies $c(G-S)\leq|S|+1$, for every non-empty subset $S$ $\subseteq V(G)$, then $G$ has the Hamiltonian path.  
\end{theorem}
\begin{proof}
	\textbf{Case 1: $|A| = |B|$}. Assume on the contrary that $G$ has no Hamiltonian path, then there exists  $u_g$ or $v_h$ in $A_2(B_2)$ that violates degree conditions $(i)$ mentioned in Theorem \ref{thm3}. Let $u_g$ be the first vertex in the ordering with $d(u_g) < g$.  By Theorem \ref{thm1}, we know that $G$ has {\em NNO} and hence $d_{G}(u_g)= g - 1$.  On removing $N(u_g)$ from $G$, we have $g$ components in $A_2$ and $A_1 \cup B_2 \cup (B_1 -N(u_g))$ forms another component.  This implies that $c(G-N(u_g)) =  g+1$.  Clearly, $g+1 \not \leq g-1$, a contradiction.\\
	
	\noindent	\textbf{Case 2: $|A| = |B|+1$}. Assume on the contrary that $G$ has no Hamiltonian path.  Let $v_{h}$ be the first vertex in $B_2$ in the ordering such that $d(v_{h}){\leq}h$. By Theorem \ref{thm1}, $G$ satisfies {\em NNO}.  Consider the set $S=B_{1}{\cup}\{v_{h+1},v_{h+2},\ldots,v_{q}\}$ and $|S|=j+q-h-1+1=j+q-h$.  Further, $c(G-S) {\geq} p+1+i-h$.  Note that $|A|=i+p$ and $|B|=j+q$.  Since $|A|=|B|+1$, $c(G-S) {\geq} p+1+i-h = j+q+1+1-h=j+q-h+2$, contradicting the premise. \qed

\end{proof}

\section{Hamiltonicity variants in $P_5$-free chordal bipartite graphs}
The variants of Hamiltonicity reported in the literature are bipancyclic, homogeneously traceable, EXACTLY 2 SIMPLE PATH COVER, Hamiltonian connected, and hypohamitonian. Interestingly, for all of them,  similar to the Hamiltonian cycle we do not know necessary and sufficient conditions. In this paper, we establish necessary and sufficient conditions for a few of the variants restricted to $P_5$-free chordal bipartite graphs.
\subsection{Bipancyclic $P_5$-free chordal bipartite graphs}
Bondy \cite{bondy1971pancyclic} has given a sufficient condition for a graph to be pancyclic. Similarly, Amar \cite{bipancyclic} has given a sufficient condition for a graph to be bipancyclic in terms of vertex degree along with Hamiltonicity. It is important to highlight that for $P_5$-free chordal bipartite graphs, being a Hamiltonian is necessary and sufficient to be bipancyclic.
\begin{theorem}
	For a $P_5$-free chordal bipartite graph $G$ of order $2n$, $G$ is Hamiltonian if and only if $G$ is bipancyclic.
\end{theorem}
\begin{proof}
		{\em Necessity:} We know that the graph $G$ satisfies Theorem \ref{thm2}. From the structural results, it is clear that the graph induced on $A_1 \cup B_1$, say $H$ is a complete bipartite subgraph and $A_2 (B_2)$ follows {\em NNO}. Let $p=|A_2|$. We obtain the cycles of length $2l$,  ~$2\leq l\leq |A_2|$, by following the Hamiltonian cycle procedure given in Theorem \ref{thm2}. \\
		$C_4=(y_1,u_1,y_2,u_2,y_1)$\\
		$C_6=(y_1,u_1,y_2,u_2,y_3,u_3,y_1)$\\
		\vdots\\
		$C_{2p}=(y_1,u_1,y_2,u_2,\ldots, u_p,y_1)$\\
		For the cycles, $C_{2l}$, $p+1\leq l\leq n$, the following procedure is used. \\
		 $C_{2(p+1)}=(y_1,u_1,y_2,u_2,\ldots,u_p,y_{p+1},x_1, y_1)$\\ $C_{2(p+2)}=(y_1,u_1,y_2,u_2,\ldots,u_p,y_{p+1},x_1,v_1,x_2,y_1)$\\
		 \vdots\\
		 $C_{2n}=(y_{1},u_{1},y_{2},u_{2},\ldots,y_{p},u_{p},y_{p+1},x_{1},v_{1},x_{2},v_{2},\ldots,x_{q},v_{q},x_{q+1},y_{p+2},x_{q+2},$ $\ldots,y_{j},x_{i},y_{1})$\\
		 Thus we obtain the cycles  $C_{2l}$, $2\leq l\leq n$ in $G$. Hence $G$ is a bipancyclic graph.\\
		 {\em Sufficiency: } Since $G$ is a $P_5$-free chordal bipartite bipancyclic graph, we have the cycles of length $2l$, $2\leq l \leq n$. It is easy to see that the cycle $C_{2n}$ contains all the vertices of $G$, which is a Hamiltonian cycle. \qed   
	\end{proof}
Remark : The above proof is constructive and we can get all these cycles in polynomial time.
\subsection{Homogeneously traceable $P_5$-free chordal bipartite graphs}
 A graph $G$ is said to be homogeneously traceable if there exists a Hamiltonian path beginning at every vertex of $G$. It is obvious that every Hamiltonian is homogeneously traceable. On the other hand, there exist homogeneously traceable non-Hamiltonian graphs. The Petersen graph is an example of a homogeneously traceable non-Hamiltonian graph.  A graph is semi-Hamiltonian if it has a Hamiltonian path and does not have a Hamiltonian cycle.  We denote the semi-Hamiltonian problem as ONE ENDPOINT SPECIFIED SEMI-HAMILTONIAN if one endpoint is specified in the input. ONE ENDPOINT SPECIFIED SEMI-HAMILTONIAN \cite{nguyen2018various} problem is NP-complete on general graphs as there is a polynomial-time reduction from satisfiability problem. The Homogeneously traceable problem is a generalization of ONE ENDPOINT SPECIFIED SEMI-HAMILTONIAN problem. Chartrand et al.\cite{chartrand1979homogeneously} has given a construction of homogeneously traceable non-Hamiltonian graph with $n$ vertices, $n\geq 9$.  Interestingly, every homogeneously traceable $P_5$-free chordal bipartite graphs are Hamiltonian graphs.
 \begin{theorem}
 	\label{thmhomogeneously}
For a $P_5$-free chordal bipartite graph $G$, $G$ is Hamiltonian if and only if $G$ is homogeneously traceable.
 \end{theorem}
\begin{proof}
		{\em Necessity:} It is easy to see that the Hamiltonian graphs are homogeneously traceable.\\
		{\em Sufficiency:} Let $G$ be homogeneously traceable. Case 1: $|A|=|B|$. By Theorem \ref{thm1}, the vertices of $A_2 (B_2)$ follows {\em NNO} property. Since $G$ is homogeneously traceable, we have Hamiltonian path beginning at every vertex of $G$. Without out loss of generality, we begin the Hamiltonian path at vertex $y_1\in B_1$. We observe that the degree of $u_1$ is at least two. In general, $\forall u_g, d(u_g) > g$, $1 \leq g \leq p$ . Similarly, if we begin at $x_1 \in A_1$, then the vertices of $B_2$ satisfies the following property, $ \forall v_h, d(v_h) >h$, $1 \leq h \leq q$. Now we observe that the graph $G$ follows Theorem \ref{thm2}. Therefore, $G$ is Hamiltonian.\\ 
		Case 2: $|A|=|B|+1$. We now show that this case is not possible. Assume on the contrary that this case is possible. Suppose we begin constructing the path $P$ at vertex $y_1\in B_1$ and visit all of $A_2$. Let $B'_1=\{y_{1},y_{2},\ldots,y_{p+1}\}$ be the set of vertices that helps to visit $A_2$. Note that $B_{3}=\{y_{p+2},\ldots,y_{j-1},y_{j}\}=B_{1}{\setminus}\{y_{1},y_{2},\ldots,y_{p+1}\}$. Similarly,   $A_{3}=\{x_{q+2},\ldots,x_{i-1},x_{i}\}=A_{1}{\setminus}\{x_{1},x_{2},\ldots,x_{q+1}\}$. Further, $|A_{3}| = | A | -(| A_{2} | +q+1)= | A | -(p+q+1)$ and $ | B_{3} | = | B | -( | B_{2} | +p+1)= | B | -(q+p+1)$. This implies that $|A_3|>|B_3|$. It is clear that some vertices of $A_3$ are left out in the path $P$. Therefore, the constructed path $P$ is not a Hamiltonian path. Hence this case is not possible.
		\qed     
		\end{proof}
\begin{corollary}
	\label{nonhamil}
	Let $G$ be a $P_5$-free chordal bipartite graph with $|A|=|B|+1$. Then, $G$ is non-Hamiltonian connected.
	\end{corollary}
\begin{proof}
Follows from the definition of Hamiltonian connected and Theorem \ref{thmhomogeneously}.\qed
	\end{proof}
\subsection{EXACTLY 2 SIMPLE PATH COVER in $P_5$-free chordal bipartite graphs}
Simple path cover is a simple path that covers all the vertices of $G$.  By EXACTLY 2 SIMPLE PATH COVER \cite{nguyen2018various}, we denote the set of graphs that can be covered by two simple paths, but cannot be covered by one simple path. A Hamiltonian path can be viewed as one simple path that covers all the vertices of $G$. These problems are NP-complete on general graphs. Since the Hamiltonian path problem is polynomial-time solvable in $P_5$-free chordal bipartite graphs, it is natural to look at the complexity of EXACTLY 2 SIMPLE PATH COVER in this graph class. The notation $\exists!z$ refers to there exists unique $z$.
\begin{theorem}
	For a $P_5$-free chordal bipartite graph $G$, $G$ has Exactly 2 Simple Path Cover if and only if one of the following is true \\
	(i)  $|A|=|B|$ and $A_2$ has an ordering, $\exists!z_r, ~d(z_r) < r$ and $\forall u_g\ne z_r, d(u_g) \geq g$, $1 \leq g \leq p$ and $B_{2}$ has an ordering, $ \forall v_h, d(v_h) \geq h$, $1 \leq h \leq q$. \\
	(ii) $|A|=|B|$ and $A_2$ has an ordering, $\forall u_g\, d(u_g) \geq g$, $1 \leq g \leq p$ and $B_{2}$ has an ordering, $\exists! z_r, ~d(z_r) < r$ and $ \forall v_h\ne z_r, d(v_h) \geq h$, $1 \leq h \leq q$. \\
	(iii) $|A|=|B|+1$ and $A_2$ has an ordering, $\exists! z_r, ~d(z_r) < r$ and $\forall u_g\ne z_r, d(u_g) \geq g$, $1 \leq g \leq p$ and $B_{2}$ has an ordering, $ \forall v_h, d(v_h) >h$, $1 \leq h \leq q$. \\
	(iv) $|A|=|B|+1$ and $A_2$ has an ordering, $\forall u_g,, d(u_g) \geq g$, $1 \leq g \leq p$ and $B_{2}$ has an ordering, $\exists! z_r, ~d(z_r) \leq r$ and $ \forall v_h\ne z_r, d(v_h) > h$, $1 \leq h \leq q$. 	
\end{theorem}
\begin{proof}
	{\em Necessity:} (i) Without loss of generality, on the contrary, assume that   a vertex $z_r\in A_2$ such that $d(z_r)<r$  or there exist at least two vertices $z_r, z_s\in A_2$ such that $d(z_r)<r$ and $d(z_s)<s$. \\
	Case 1: There does not exist $z_r$ such that $d(z_r)<r$. By Theorem \ref{thm3}, we observe that $G$ is an yes instance of the Hamiltonian path problem, which is a one simple path cover. A contradiction. \\
	Case 2: There exist at least two such vertices $z_r$ and $z_s$. Without loss of generality, assume that $d(z_r) \leq d(z_s)$. By following the procedure given in Theorem \ref{thm3}, we obtain three simple paths $P_1$, $P_2$, and $P_3$ that cover  $V(G)$. \\
	$P_1=(u_{1},y_{1},u_{2},y_{2},\ldots,u_{r-1}, y_{r-1}, u_{r}=z_r)$, $1<r<s$ \\ 
	$P_2=(y_{r}, u_{r+1}, y_{r+1}, u_{r+2},\ldots,y_{s-2}, u_{s-1}, y_{s-1}, u_{s}=z_s)$, $1<s\leq p$ \\
	$P_3= (y_{s}, u_{s+1}, \dots,y_{p-1}, u_{p},y_{p},x_{q+1},y_{p+1},x_{q+2},y_{p+2},\ldots,x_{i}, y_{j},x_{q}, v_{q}, \ldots,x_{2},$ $v_{2},x_{1}, v_{1})$, \\
	A similar argument can be given for (ii), (iii) and (iv)\\
	{\em Sufficiency:} There exists exactly one vertex $z_r$, such that $~d(z_r) < r$. By Theorem \ref{thm3}, $G$ is a no instance of the Hamiltonian path problem and thus $G$ cannot be covered by one simple path.  By following the procedure given in  Theorem \ref{thm3}, we obtain the following two simple paths $P_1$, and $P_2$ that cover $V(G)$.  \\
	$P_1=(u_{1},y_{1},u_{2},y_{2},\ldots,u_{r-1},y_{r-1},u_{r}=z_r)$,  $1<r\leq p$\\
	$P_2=(y_{r}, u_{r+1}, y_{r+1},u_{r+2}\ldots,y_{p-1},u_{p},y_{p},x_{q+1},y_{p+1},x_{q+2},y_{p+2},\ldots,x_{i}, y_{j},x_{q}, v_{q}, \ldots,x_{2},$ $v_{2},x_{1},v_{1})$\\
	Similarly (ii), (iii) and (iv) can be proved.\\ 
\qed 
\end{proof}

\subsection{Hamiltonian connected $P_5$-free chordal bipartite graphs}
A graph $G$ is said to be Hamiltonian connected if there exists a Hamiltonian path between every pair of vertices. 
Some of the related problems studied in the literature are TWO ENDPOINT SPECIFIED-SEMI-HAMILTONIAN \cite{nguyen2018various}, where both the endpoints are specified. This problem is NP-complete on chordal bipartite graphs. Hamiltonian connected is a generalization of this problem.  In this paper, we investigate the Hamiltonian connectedness of $G$. If $G$ has a Hamiltonian path, then $G$ satisfies Theorem \ref{thm3}. It is clear from Corollary \ref{nonhamil} that if $|A|=|B|+1$, then $G$ is non-Hamiltonian connected. We shall now look at the case where $|A|=|B|$.
\begin{theorem}
	Let $G$ be a $P_5$-free chordal bipartite graph with $|A|=|B|$. Then, $G$ is non-Hamiltonian connected.
\end{theorem}
\begin{proof}
		Assume on the contrary, that $G$ is Hamiltonian connected. By the definition, there exists a Hamiltonian path between every pair of vertices. Without loss of generality, we choose a pair $u_1$ and $u_2$ from $A_2$. Suppose, if we begin constructing the path $P$ at $u_1$ and visit the vertices of $A_2$ using $B_1$ except the vertex $u_2$. 
		 Note that $B_{3}=\{y_{p},\ldots,y_{j-1},y_{j}\}=B_{1}{\setminus}\{y_{1},y_{2},\ldots,y_{p-1}\}$ and  $A_{3}=\{x_{q+2},\ldots,x_{i-1},x_{i}\}=A_{1}{\setminus}\{x_{1},x_{2},\ldots,x_{q+1}\}$. Further, $|A_{3}| = | A | -(| A_{2} | +q+1)= | A | -(p+q+1)$ and $ | B_{3} | = | B | -( | B_{2} | +p-1)= | B | -(q+p-1)$. 
	Since $ | A | = | B | $, it follows that $ | A_{3} | < | B_{3} | $, in particular, $ | B_{3} |=| A_{3} |+2 $. Now, we are left with two vertices in $B_1 $, say $w$ and $z$, and a vertex $u_2$.
	 We observer that either $P$ ends in $w$ $(z)$ leaving $u_2$ or $P$ ends in $u_2$ leaving $w$ $(z)$. Clearly, $P$ is not a Hamiltonian path, which is a contradiction.  Hence $G$ is non-Hamiltonian connected. \qed
	\end{proof} 
\subsection{Path-hypohamitonian $P_5$-free chordal bipartite graphs}
We know that bipartite graphs are non-hypohamiltonian graphs. So it is natural to ask for the variants of hypohamiltonian, one such variant is path-hypohamiltonian. A graph $G$ is called path-hypohamiltonian,  if $G$ has no Hamiltonian path and $\forall v \in V(G)$, $G - v$ has a Hamiltonian path. In this section, we show that $P_5$-free chordal bipartite graphs are non-path-hypohamiltonian.

\begin{theorem}
	Let $G$ be a $P_5$-free chordal bipartite graph and has no Hamiltonian path. Then, $G$ is non-path-hypohamiltonian.
\end{theorem}
\begin{proof}
	We now exhibit a vertex $u$ such that $G-u$ has no Hamiltonian path. Since $G$ is a no instance of Hamiltonian path problem, by Theorem \ref{thm3}, either (i) or (ii) is true. (i) $|A|=|B|$ and there exists a vertex $u_g\in A_2$ such that $d(u_g)<g$ or $v_r\in B_2$ such that $d(v_r)<r$ (ii) $|A|=|B|+1$ there exists a vertex $u_g\in A_2$ such that $d(u_g)<g$ or $v_r\in B_2$ such that $d(v_r)\leq r$.\\
	(i) We observe that $G-u_g$ is an yes instance of the Hamiltonian path problem and thus we get a Hamiltonian path in $G-u_g$. Now we choose $v \in B_2$.  In $G-v$, the degree $d_G(u_g)=d_{G-v}(u_g)$. By Theorem \ref{thm3},  $G-v$ is a no instance of the Hamiltonian path problem. Therefore, $G$ is non-path-hypohamiltonian.\\
	(ii) We choose $v \in B_1 (B_2)$. In $G-v$, It is easy to see that $|A|>|B|+1$. By Theorem \ref{thm3}, $G-v$ is a no instance of Hamiltonian path problem. Therefore, $G$ is non-path-hypohamiltonian.
	

	
	\end{proof}
{\bf Remark:} The problems discussed in this section and their proofs are constructive in nature, therefore we obtain a polynomial-time algorithm for Hamiltonicity variants.

\section{Generalizations of Hamiltonian path (cycle) in $P_5$-free chordal bipartite graphs}
We know that, if a problem $X$ is NP-complete for a graph class under study, then the problem  $X$ is also NP-complete in all its superclasses. Further, the generalization of the problem under study is also NP-complete in that graph class. For example, the Hamiltonian path is NP-complete on split graphs therefore it is NP-complete in chordal graphs. The longest path problem which is a generalization of the Hamiltonian path problem is NP-complete in split graphs. If a problem is polynomial-time solvable in a graph class, then it is appropriate to investigate the complexity of its generalization in that graph class. Interestingly, in $P_5$-free chordal bipartite graphs, Hamiltonian path (cycle) problems are polynomial-time solvable. So it is natural to look at the generalization of the Hamiltonian path (cycle) problems and their complexity study in $P_5$-free chordal bipartite graphs. In this section, we use the Hamiltonian cycle (path) problem as a framework to solve other combinatorial problems. For all other problems, we modify the input graph to obtain  $G^*$. The challenge lies in identifying $G^*$ for each combinatorial problem such as longest cycle (path), Steiner cycle (path). By calling the appropriate algorithm (Hamiltonian cycle, or Hamiltonian path Algorithm) on $G^*$, we obtain a result. A suitable modification to the result will give the longest cycle (path), Steiner cycle (path) for $G$. We shall see some of the generalizations of Hamiltonicity.

\subsection{Longest paths in $P_5$-free chordal bipartite graphs}
For a connected graph $G$, the longest path is an induced path of maximum length in $G$.  Since the Hamiltonian path is a path of maximum length, finding the longest path is trivially solvable if the input instance is an yes instance of the Hamiltonian path problem.  Thus the longest path problem is a generalization of the Hamiltonian path problem, and hence the longest path problem is NP-complete if the Hamiltonian path problem is NP-complete on the graph class under study.  On the other hand, it is interesting to investigate the complexity of the longest path problem in graphs where the Hamiltonian path problem is polynomial-time solvable.   Since, the Hamiltonian path problem in $P_{5}$-free chordal bipartite graphs is polynomial-time solvable, in this section, we shall investigate the complexity of the longest path problem in $P_{5}$-free chordal bipartite graphs.  \\\\
{\bf Pruning:} We shall now prune $G$ by removing vertices that will not be part of any longest path in $G$.  Without loss of generality, we assume that $G$ has no Hamiltonian path, and hence $A=B+1+f$, $f \geq 1$ or there must exists a vertex in $A_2$ ($B_2$) that does not satisfies the conditions mentioned in Theorem \ref{thm3}.  As part of pruning, we prune such vertices from $G$.   Recall that $A_2=(u_{1},u_{2},\ldots,u_{p})$.  Let  $u_r$ be the first vertex in $A_2$ with $d(u_{r})<r$.  Remove $u_r$ and relabel the vertices of $A_2$ so that the sequence is reduced to $(u_{1},u_{2},\ldots,u_{p-1})$.  With respect to the modified sequence, if we find $u_i$ such that $d(u_{i}) < i$, then prune $u_i$ and update the sequence.  If there are no such $u_r$, then $c=0$.  After, say $c$ iterations, $A_2$ becomes $(u_{1},u_{2},\ldots,u_{p-c})$ such that for $\forall u_g, 1 \leq g \leq (p-c), d(u_g) \geq g$.  Similarly, after $d$ iterations, $B_2$ becomes $(v_{1},v_{2},\ldots,v_{q-d})$ such that for $\forall v_h, 1 \leq h \leq (q-d), d(v_h) \geq h$.   After pruning the vertices in $A$ is reduced to the set $A'$, $|A'|= |A| - c$,, similarly, $B$ reduced to the set $B'$, $|B'|=|B|-d$.  From now on, when we refer to $A_2$ ($B_2$), it refers to the modified $A_2$ ($B_2$).  Let $G^*$ be the modified graph of $G$. $V(G^*) = V(A') \cup V(B')$. \\

\noindent \textbf{Case 1: $|A'| = |B'|$}\\
We observe that $G^*$ satisfies {\em NNO} and as per Theorem \ref{thm3}, $G^*$ has a Hamiltonian path, which is $ P= (u_1, y_1, u_2, y_2, \ldots, u_{p-c}, y_{p-c}, x_{(q-d)+1}, y_{(p-c)+1}, $ $x_{(q-d)+2}, y_{(p-c)+2}, \ldots, x_i, y_j, x_{q-d}, v_{q-d},$ $x_{(q-d)-1},$ $v_{(q-d)-1},$  $ \ldots, x_1, v_1)$  \\
\textbf{Case 2: $|A'| = |B'|+1+f$, $f \geq 0 $} \\
If $f>0$, By Theorem \ref{thm3}, $G^*$ is not an yes instance of the Hamiltonian path problem.  We remove $f$ vertices from $A'_2$.  Let $G^*_1(A' \setminus \{u_{(p-c)-1},u_{(p-c)-2},\ldots, u_{(p-c)-f}\},B')$ be the modified graph.  If $f=0$, then $G^*_1(A',B')$ be same as $G^*$. \\
Case 2.1: $\exists  v_r$ $\in B'_2 $ in $G^*_1$ such that $d_{G^*_1}(v_{r}) = r$\\
Since $d_{G^*_1}(v_{r}) = r$, $G^*_1$ is not an yes instance of the Hamiltonian path problem as per Theorem \ref{thm3}.  However $G^*_1(A' \setminus \{u_{(p-c)-1},u_{(p-c)-2},\ldots, u_{(p-c)-f},u_{(p-c)-(f+1)}\},B')$ has a Hamiltonian path as per Theorem \ref{thm3}.\\
$ P= (u_1, y_1, u_2, y_2, \ldots, u_{(p-c)-(f+1)}, y_{(p-c)-(f+1)}, x_{(q-d)+1},y_{((p-c)-(f+1))+1},x_{(q-d)+2}, y_{((p-c)-(f+1))+2}, \ldots, x_i,$ $ y_j, x_{q-d}, v_{q-d},$ $ x_{(q-d)-1},$ $ v_{(q-d)-1}, \ldots,$ $ x_1, v_1)$  \\ 
Case 2.2: $\nexists  v_r$ $\in B'_2 $ in $G^*_1$ such that $d_{G^*_1}(v_{r}) = r$\\ 
By Theorem \ref{thm3}, $G^*_1$ is an yes instance of the Hamiltonian path problem.\\
$ P= (u_1, y_1, u_2, y_2, \ldots, u_{(p-c)-f}, y_{(p-c)-f}, x_{(q-d)+1},y_{((p-c)-f)+1},x_{(q-d)+2}, y_{((p-c)-f)+2}, \ldots, x_i, y_j, x_{q-d}, v_{q-d},$ $ x_{(q-d)-1},$ $ v_{(q-d)-1}, \ldots,$ $ x_1, v_1)$  \\ 

\noindent\textbf{Claim 1:} $P$ is a longest path in $G^*$. 
\begin{proof}
	$G^*$ be the graph obtained by pruning the vertices from $A_2$ and $B_2$ in $G$ and thus $G^*$ becomes an yes instance of the Hamiltonian path problem. Since $G^*$ satisfies {\em (NNO)}, for all the pruned vertices of $A_2$ and $B_2$, their neighborhood is a subset of $\{y_1,\ldots,y_{p-c}\}$ and $\{x_1,\ldots,x_{q-d}\}$ respectively. This shows that the pruned vertices of $A_2$ and $B_2$ cannot be augmented to $P$ to get a longer path in $G$. This proves that $P$ is a longest path in $G$. \qed
\end{proof}
\noindent\textbf{Trace of the longest path algorithm:}
\begin{figure}
	\begin{center}
		\includegraphics[width=80mm,scale=0.5]{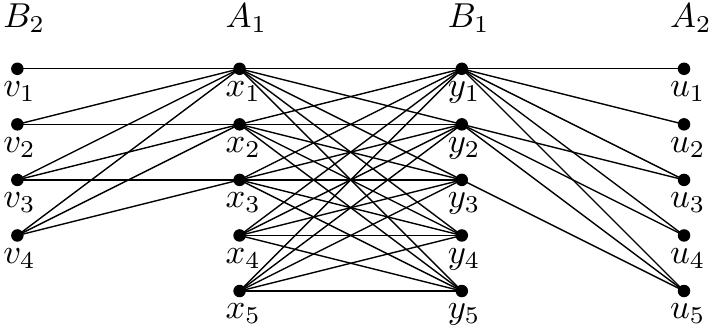}
		\caption{\small \sl   An illustration for the proof of Longest path (Case 1).\label{fig:Lp1}}
	\end{center}
\end{figure}
For Figure \ref{fig:Lp1}, we shall trace the longest path algorithm.  Consider the vertices of $A_2$,  $d(u_2)=1$, $u_2$ violates the degree constraint, so we prune $u_2$ and relabel the vertices $u_3$ as $u_2$, $u_4$ as $u_3$, and $u_5$ as $u_4$.  The updated sequence of $A_2$ is $(u_1,u_2,u_3,u_4)$. 
With respect to the modified sequence, $d_{G^*}(u_3)=2$, prune $u_3$ and relabel the vertex $u_4$ as $u_3$. Now all the vertices of $A'_2$ satisfies the degree constraint and the sequence is $(u_1,u_2,u_3)$.  
By applying the procedure to the vertices of $B_2$, we get $B'_2$ as $(v_1,v_2,v_3)$.  
The resultant graph $G^*$ falls under Case 1.  We obtain the longest path $ P= (u_1, y_1, u_2, y_2, u_3, y_3, x_4, y_4,x_5, y_5,x_3,v_3,x_2,v_2,x_1,v_1)$  \\
\begin{figure}
	\begin{center}
		\includegraphics[width=80mm,scale=0.5]{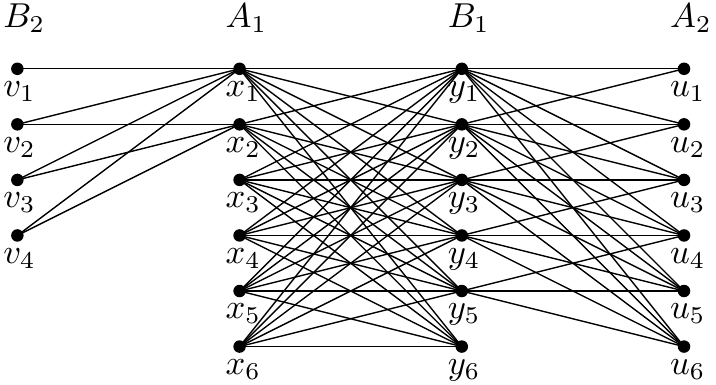}
		\caption{\small \sl  An illustration for the proof of Longest path (Case 2).\label{fig:Lp2}}
	\end{center}
\end{figure}

\noindent Consider Figure \ref{fig:Lp2}.  The vertex $u_6$ violates the degree constrains, prune $u_6$ and we obtain $A'_2$, whose sequence reduced to $(u_1,u_2,u_3,u_4,u_5)$.  Similarly $v_3$ and $v_4$ violates the constraint and thus the sequence of $B'_2$ is reduced to  $(v_1,v_2)$.  $G^*$ follows Case 2 of the procedure. Let $G^*_1$ be the graph obtained by removing $f$,$f=\{u_4,u_5\}$ vertices from $G^*$.  We observe that $d_{G^*}(v_1)=1$, by Case 2.1, remove $u_3$ and the longest path is  $ P= (u_1, y_1, u_2, y_2, x_3, y_3, x_4, y_4,x_5, y_5,x_6, y_6,x_2,v_2,x_1,v_1)$  \\ 
\begin{theorem}
	\label{thmlong}
	Let $G$ be a $P_5$-free chordal bipartite graph.  Finding the longest path in $G$ is polynomial-time solvable.
\end{theorem}
\begin{proof}
	Follows from the above discussion on pruning and Claim 1. \qed
\end{proof}
{\bf Remarks:} As an extension of longest path problem, we naturally obtain a minimum leaf spanning tree of $G$, which is a spanning tree of $G$ with the minimum number of leaves, in polynomial-time.  Since $G$ satisfies {\em {\em {\em NNO}}}, the vertices pruned while constructing $G^*$, cannot be included as internal vertices of $P$.  We shall now construct a minimum leaf spanning tree $T$ with $P$ as a subtree.  (i) the pruned vertices are augmented to $P$ as leaves to obtain $T$.  (ii) if $|A_1'|>|B_1'|$, then the vertices in $A_1'$ not included in $P$ are added to $P$ as leaves to obtain $T$.  Similarly, if $|B_1'|>|A_1'|$, then the vertices in $B_1'$ not included in $P$ are added to $P$ as leaves to obtain $T$.   This shows that $T$ has $|V(G) \setminus V(P)| + 2$ leaves.  Maximum leaf spanning tree of $G$ can be constructed by choosing $x_1y_1$ edge and augment all other vertices of $A_1, A_2$ to the vertex $y_1$, similarly $B_1,B_2$ to $x_1$. 
\begin{corollary}
	Minimum connected dominating set in $P_5$-free chordal bipartite graph is polynomial-time solvable.
	\end{corollary}
\begin{proof}
	Follows from Theorem \ref{thmlong}.
	\end{proof}
\subsection{Longest cycles in $P_5$-free chordal bipartite graphs}
Similar to the longest path, the longest cycle is an induced cycle of maximum length in $G$. It is natural to ask for the longest cycle in a graph if $G$ has no Hamiltonian cycle because it is trivially solvable when $G$ has a Hamiltonian cycle. Since, the Hamiltonian cycle problem in $P_5$-free chordal bipartite graphs is polynomial-time solvable, in this section, we shall investigate the complexity of the longest cycle problem.
\\ \textbf{Pruning}: Without loss of generality, we assume that $G$ has no Hamiltonian cycle, and hence $A\ne B$ or there must exist a vertex in $A_2 ~(B_2)$ that does not satisfies the condition mentioned in Theorem 2.  The violated vertices cannot be a part of the longest cycle, so we prune such vertices from $G$.  Let $u_r$ be the first vertex in $A_2$ with $d(u_r) \leq r$.  Remove $u_r$ and relabel the vertices of $A_2$ until there are no such  $u_r$.  After say $c$ iterations, $A_2$ becomes $(u_1, u_2,\ldots, u_{p-c})$ such that for $\forall$$u_g$, 1$\leq g \leq (p-c)$, $d(u_g) > g$. 
Similarly, after say $d$ iterations, $B_2$ becomes $(v_1, v_2,\ldots, v_{q-d})$ such that for $\forall$$v_h$, 1$\leq h \leq (q-d)$, $d(u_h) > h$. After pruning, $A$ reduced to the set $A'$, $|A'|= |A|- c$ and $B$ reduced to the set $B'$, $|B'|=|B|-d$.  Let the modified graph be $G^*$.\\ 

\noindent \textbf{Case 1: $|A'| = |B'|$}\\
$G^*$ satisfies {\em NNO} and by Theorem \ref{thm2}, $G^*$ is an yes instance of the Hamiltonian cycle problem. \\
$ C= (y_1,u_1, y_2,u_2,  \ldots, y_{p-c},u_{p-c}, y_{(p-c)+1}, x_1, v_1,x_2, v_2, \ldots,$ $  x_{q-d},$ $ v_{q-d}, x_{(q-d)+1},y_{(p-c)+2}, x_{(q-d)+2},  \ldots,$ $  y_j,$ $ x_i, y_1 )$ \\ 
\noindent \textbf{Case 2: $|A'| = |B'|+f$, $f\geq 1$}\\
By Theorem \ref{thm2}, $|A'| = |B'|$ and hence $G^*$ is an yes instance of the Hamiltonian cycle problem.  We Remove $f$, $\{u_{(p-c)-1},u_{(p-c)-2},\ldots, u_{(p-c)-f}\}$ vertices from $A'_2$ in $G^*$ results  $|A'| = |B'|$ and by Theorem \ref{thm2}, $G^*$ has a Hamiltonian cycle.
\\ $ C= (y_1,u_1, y_2,u_2,  \ldots, y_{(p-c)-f},u_{(p-c)-f}, y_{((p-c)-f)+1}, $ $x_1, v_1,x_2, v_2,$ $\ldots,x_{q-d},$ $ v_{q-d}, x_{(q-d)+1},y_{((p-c)-f)+2},$ $ x_{(q-d)+2},  \ldots,$ $ y_j, x_i, y_1 ) $ \\

\noindent \textbf{Claim 2:} $C$ is a longest cycle in $G^*$. 
\begin{proof}
	We prune the vertices from $A_2$ and $B_2$ in $G$ and let the modified graph be $G^*$.  We observe that $G^*$ satisfies the condition mentioned in Theorem 2.  Since $G^*$ satisfies {\em (NNO)}, the neighborhood of pruned vertices of $A_2$ and $B_2$ is a subset of $\{y_1,\ldots,y_{p-c}\}$ and $\{x_1,\ldots,x_{q-d}\}$ respectively.  Therefore the pruned vertices cannot be augmented to $C$ to get a longer cycle in $G^*$.  Hence $C$ is a longest cycle in $G^*$. \qed
\end{proof}

\noindent\textbf{Trace of the Algorithm:}
Consider the graph $G$ given in Figure \ref{fig:Lp1}, $u_1$ and $u_2$ violates the degree constraint, we prune the vertices and the sequence becomes $(u_1,u_2,u_3)$.  With respect to the modified sequence $d_{G^*}(u_3)=3$, prune $u_3$ and the sequence becomes $(u_1,u_2)$.  Similarly, in $B_2$, $v_1$ is pruned, $(v_1,v_2,v_3)$.  We find $v_3$ such that $d_{G^*}(v_3)=3$, prune $v_3$, $(v_1,v_2)$.  Graph $G^*$ satisfies case 1, $|A'| = |B'|$ and the longest cycle $C=(y_1,u_1, y_2,u_2,y_3, x_1, v_1,x_2, v_2,x_3,y_4,x_4,y_5,x_5,y_1)$. \\
Case 2: Consider the graph $G$ given in Figure \ref{fig:Lp2}, vertices $u_5$ and $u_6$ are pruned and $A'$ becomes  $(u_1,u_2,u_3,u_4)$.  Similarly, $v_1$ is pruned and the sequence becomes $(v_1,v_2,v_3)$.  In the modified graph, $v_2$ and $v_3$ violates degree constraint and thus $B'_2$ has a vertex $v_1$.  $G^*$ satisfies case 2 $|A'| = |B'|+f$, $f=(u_2,u_3,u_4)$.  Let $G^*_1$ be the graph obtained by removing $(u_2,u_3,u_4)$.  The longest cycle $C=(y_1,u_1, y_2, x_1, v_1,x_2,y_3,x_3,y_4,x_4,y_5,x_5,y_1)$. 
\subsection{Steiner paths in $P_{5}$-free chordal bipartite graphs}
The Steiner path problem is introduced in \cite{Steinerpathcograph} which is a variant of the Steiner tree and a generalization of the Hamiltonian path problem.  Given $G, R \subseteq V(G)$, find a path containing all of $R$ minimizing $V(G) \setminus R$, if exists.  Note that, this {\em constrained path problem} is precisely the Hamiltonian path problem when $R=V(G)$.  This has another motivation as well.  The Steiner tree problem \cite{mullar1987brandstadt} is the problem of connecting a given set $R$ of vertices (known as terminal vertices) by adding a minimum number of vertices from $V(G) \setminus R$ (known as Steiner vertices).  The Steiner tree problem is trivially solvable in $P_{5}$-free chordal bipartite graphs.  It is natural to ask for a variant, namely the Steiner path.  It is important to highlight that not all input graphs have Steiner paths.  We shall present constructive proof for the existence of the Steiner path in a $P_{5}$-free chordal bipartite graph.
\begin{lemma}
	\label{lemma spath rca2}
	For a $P_5$-free chordal bipartite graph and $R=\{u_1,\ldots,u_r\} \subseteq A_2$, $G$ has Steiner path $P$ if and only if the vertices of $R$ has an ordering $(u_1,\ldots,u_r)$ such that $\forall u_g, d(u_g) \geq g$, $1 \leq g \leq r-1$ and $d(u_r) \geq d(u_{r-1})$. 
\end{lemma}
\begin{proof} 
	We modify the instance of Steiner path problem of $G$ to the instance of Hamiltonian path problem of $G^*$, where $G^*$ is the graph induced on the set $R \cup \{y_1,y_2,\ldots,y_{r-1}\}$.\\
	{\em Necessity :} Assume on the contrary, there exists a vertex $u_g\in R$ such that $d(u_g) < g$. Since $G^*$  is a no instance of the Hamiltonian path problem, the path $P$ obtained from $G^*$ is not a Steiner path, a contradiction.\\
	{\em Sufficiency:}  Since $R$ has an ordering $(u_1,\ldots,u_r)$ and $\forall u_g, d(u_g) \geq g$, $1 \leq g \leq r-1$, $R$ satisfies {\em NNO} in $G^*$. By passing $G^*$ to the Hamiltonian path algorithm we obtain a path $P$. Since the path $P$ spans all of $R$, $P$ is a Steiner path in $G$. Therefore, $P=(u_1,y_1,\ldots,y_{r-1},u_r)$ is a Steiner path.  Since $R$ is an independent set of size $r$, any path containing $R$ must have $r-1$ additional vertices and hence $P$ is a minimum Steiner path.  
	\qed
\end{proof}
\begin{lemma}
	\label{lemma spath rca1}
	For a $P_5$-free chordal bipartite graph and $R=\{x_1,\ldots,x_r\} \subseteq A_1$, $G$ has Steiner path $P$ if and only if  there exists $(v_1,\ldots,v_{r-(j+1)})$ in $B_2$ such that $\forall v_h, d(v_h) \geq h$, $1 \leq h \leq |A_1|- |B_1|$.
	\end{lemma}
\begin{proof}
	 Let $G^*$ be the graph induced on  $R\cup B_1\cup\{v_1,\ldots,v_{r-(j+1)}\}$.\\
	{\em Necessity :} 	 It is easy to see that if $|R| \leq |B_1|+1$, then $P=(x_1,y_1,\ldots,x_{r-1},y_{r-1},x_r)$ is a minimum Steiner path. 
	 We shall now see the case where, $|R| > |B_1|+1$. Assume on the contrary, there exists a vertex $v_h\in B_2$ such that $d(v_h) < h$. We observe that $G^*$ is a no instance of Hamiltonian path problem, a contradiction. \\
	{\em Sufficiency:} We observe that $G^*$ has a Hamiltonian path which is a Steiner path in $G$. Note that the path
	 $P=(x_r,y_j,x_{r-1},y_{j-1},\ldots,x_{r-(j+1)},v_{r-(j+1)},\ldots, v_2,x_2,v_1,x_1)$ is a Steiner path of minimum cardinality. \qed
	\end{proof}
\begin{lemma}
	\label{lemma spath rca1b2}
		For a $P_5$-free chordal bipartite graph and $R=\{x_1,\ldots,x_r\} \subseteq (A_1 \cap B_2$), $G$ has Steiner path $P$ if and only if $R \cap B_2=(z_1,\ldots,z_l)$ has {\em NNO} with $(w_1,\ldots,w_l)$, $1\leq l<r$, of $A_1$.
\end{lemma}
\begin{proof}
	Let $G^*$ be the graph induced on the set $R\cup\{z_1,\ldots,z_l\}\cup \{w_1,\ldots,w_l\}.  $\\
	{\em Necessity :} 	Proof is similar to the necessity of Lemma \ref{lemma spath rca1}. \\
	{\em Sufficiency:} 	The modified graph $G^*$ has a Hamiltonian path. We obtain the path $P=(z_1,w_1,\ldots,z_l,w_l,y_1,x_1,$ $\ldots,y_{r-l-1},x_{r-l})$ which is a minimum Steiner path in $G$. \qed
\end{proof}
On the similar line, other cases $(a)$ $R \subseteq B_1$, $(b)$ $R \subseteq B_2$, $(c)$ $R \cap A_1 \not= \emptyset$ and $R \cap A_2 \not= \emptyset$, $(d)$ $R \cap A_2 \not= \emptyset$ and $R \cap B_1 \not= \emptyset$, $(e)$ $R \cap B_1 \not= \emptyset$ and $R \cap B_2 \not= \emptyset$, and $(f)$ $R \cap A_1 \not= \emptyset$ and $R \cap B_2 \not= \emptyset$ and  $R \cap A_2 \not= \emptyset$ can be proved. \\\\
\textbf{Trace:} Let us consider the set, $R=\{u_2,u_3,u_5\}$ in Figure \ref{fig:stpath} (a). It is clear that all the vertices in set $R$ satisfies the above given degree constrains and hence there exists a Steiner path, $P=\{u_2,y_1,u_3,y_2,u_5\}$. Suppose, if $R=\{u_1,u_2,u_3\}$,  then we observe that the vertex $u_2$ violates the degree constraint and hence there does not exist Steiner path in Figure \ref{fig:stpath} (b).
\begin{figure}
	\begin{center}
		\includegraphics[width=80mm,scale=0.5]{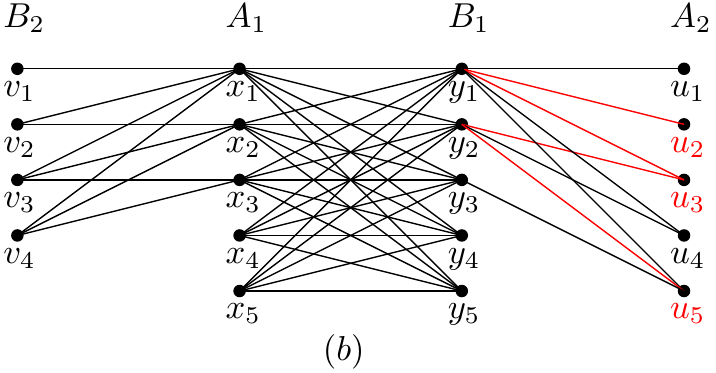}
		\includegraphics[width=80mm,scale=0.5]{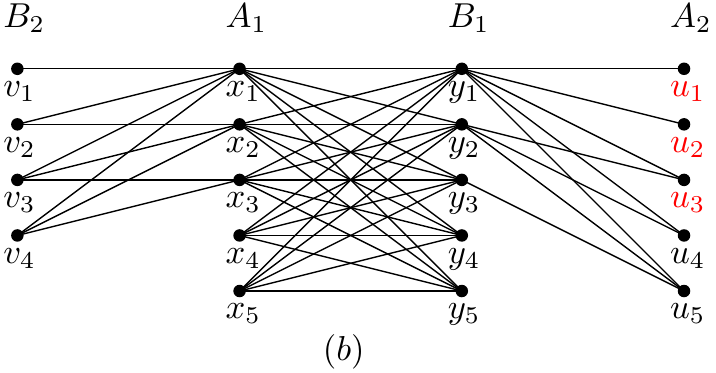}
		\caption{\small \sl  An example graph (a) Yes instance of Steiner path problem, (b) No instance of Steiner path problem.\label{fig:stpath}}
	\end{center}
\end{figure}\\

\subsection{Steiner cycles in $P_5$-free chordal bipartite graphs}
Similar to the Steiner path, given a graph $G, R \subseteq V(G)$, the Steiner cycle problem asks for a simple cycle containing all of $R$ minimizing $V(G) \setminus R$. Suppose, if $R=V(G)$, then this problem is precisely the Hamiltonian cycle problem. If $G$ is an yes instance of the Hamiltonian cycle problem, then this problem is trivially solvable.  We observe that not all input graphs have a solution to the Steiner cycle problem.  We shall present constructive proof for the existence of a Steiner cycle in a $P_5$-free chordal bipartite graph.
\begin{lemma}
	\label{lemma scycle rca2}
	For a $P_5$-free chordal bipartite graph and $R=\{u_1,\ldots,u_r\} \subseteq A_2$, $G$ has Steiner cycle $P$ if and only if the vertices of $R$ has an ordering $(u_1,\ldots,u_r)$ such that $\forall u_g, d(u_g) > g$, $1 \leq g \leq r-1$ and $d(u_r) \geq d(u_{r-1})$. 
\end{lemma}
\begin{proof} 
	 Let $G^*$ be the graph induced on the set $R\cup (u_1,\ldots,u_r)$.\\
	{\em Necessity :} Assume on the contrary, there exists a vertex $u_g\in R$ such that $d(u_g) \leq g$. Since $G^*$  is a no instance of the Hamiltonian cycle problem, the cycle $C$ obtained from $G^*$ is not a Steiner path, a contradiction.\\
	{\em Sufficiency:}  Since $R$ has an ordering $(u_1,\ldots,u_r)$ and $\forall u_g, d(u_g) > g$, $1 \leq g \leq r-1$, $R$ satisfies {\em NNO} in $G^*$. We call our Hamiltonian cycle algorithm on $G^*$ which outputs a cycle $C$. Since the cycle $C$ spans all of $R$, $C$ is a Steiner cycle in $G$. Therefore, $C=(y_1,u_1,\ldots,y_{r},u_r,y_1\}$ is a Steiner cycle.  Any cycle containing $R$ must have $r$ additional vertices and hence $C$ is a minimum Steiner cycle. 
	\qed
\end{proof}
\begin{lemma}
	\label{lemma scycle rca1}
	For a $P_5$-free chordal bipartite graph and $R=\{x_1,\ldots,x_r\} \subseteq A_1$, $G$ has Steiner cycle $C$ if and only if there exists $(v_1,\ldots,v_{r-j})$ in $B_2$ such that $\forall v_h, d(v_h) >h$, $1 \leq h \leq r-j$.

\end{lemma}
\begin{proof}
  Let $G^*$ be the graph induced on  $R\cup B_1\cup \{v_1,\ldots,v_{r-j}\}$.\\
	{\em Necessity :} 		If $|R| \leq |B_1|$ , then the Steiner cycle problem is trivially solvable.  We obtain the Steiner cycle $C=(x_1,y_1,\ldots,x_{r-1},$ $y_{r-1},x_r,y_r,x_1)$.  We assume that $|R| > |B_1|$. 
	Assume on the contrary, there exists a vertex $v_h\in B_2$ such that $d(v_h) \leq h$. We observe that $G^*$ is a no instance of Hamiltonian cycle problem, a contradiction. \\
	{\em Sufficiency:} We observe that $G^*$ has a Hamiltonian cycle which is a Steiner cycle in $G$. The cycle $C=(x_1,v_1,x_2,v_2,\ldots, x_{r-j}, v_{r-j}, x_{(r-j)+1}, y_j,x_{(r-j)+2}, y_{j-1}, \ldots, x_{r-1},y_2,x_r,y_1,x_1)$.\\
     is a Steiner cycle of minimum cardinality. \qed
\end{proof}

\begin{lemma}
	\label{lemma scycle rca1b2}
	For a $P_5$-free chordal bipartite graph and $R=\{x_1,\ldots,x_r\} \subseteq A_1 \cap B_2$, $G$ has Steiner cycle $C$ if and only if $R \cap B_2=(z_1,\ldots,z_l)$ has {\em NNO} with $(w_1,\ldots,w_l)$, $1\leq l<r$, of $A_1$.
\end{lemma}
\begin{proof}
	Let $G^*$ be the graph induced on the set $R\cup\{z_1,\ldots,z_l\}\cup \{w_1,\ldots,w_l\}.  $\\
	{\em Necessity :} 	Proof is similar to the necessity of Lemma \ref{lemma scycle rca1} \\
	{\em Sufficiency:} 	We observe that $G^*$ has a Hamiltonian cycle. The cycle  
	$C=(w_1,z_1,\ldots, w_l,z_l,x_{l+1},y_1,\ldots x_{r-l},$ $y_{r-l},w_1)$ is a minimum Steiner cycle in $G$. \qed
\end{proof}
For other cases such as $(a)$ $R \subseteq B_1$, $(b)$ $R \subseteq B_2$, $(c)$ $R \cap A_1 \not= \emptyset$ and $R \cap A_2 \not= \emptyset$, $(d)$ $R \cap A_2 \not= \emptyset$ and $R \cap B_1 \not= \emptyset$, $(e)$ $R \cap A_1 \not= \emptyset$ and $R \cap A_2 \not= \emptyset$, and $(f)$ $R \cap A_1 \not= \emptyset$ and $R \cap B_2 \not= \emptyset$ and  $R \cap A_2 \not= \emptyset$, a similar argument can be given.\\ 
\textbf{Trace:} Consider the illustration given in Figure \ref{fig:stcycle}.  We trace $(a)$ and $(b)$ of Figure \ref{fig:stcycle} for the set $R=\{x_1,x_2,x_3,x_4,x_5\}$. In Figure \ref{fig:stcycle} (a), $|B_1|=3$, so there must exist at least two vertices, $\{z_1,z_2\}$ in $B_2$ such that $d(z_r)>r$. We consider $v_2$ as $z_1$ and $v_3$ as $z_2$. We observe that $d(z_1)>1$ and $d(z_2)>2$. The Steiner cycle, $C=\{x_1,v_2,x_2,v_3,x_3,y_1,x_4,y_2,x_5,y_3,x_1\}$. In Figure \ref{fig:stcycle} (b), there does not exist at least two vertices satisfying the degree constraint and hence it is a no instance of a Steiner cycle.\\
\begin{figure}
	\begin{center} 
		\includegraphics[width=80mm,scale=0.5]{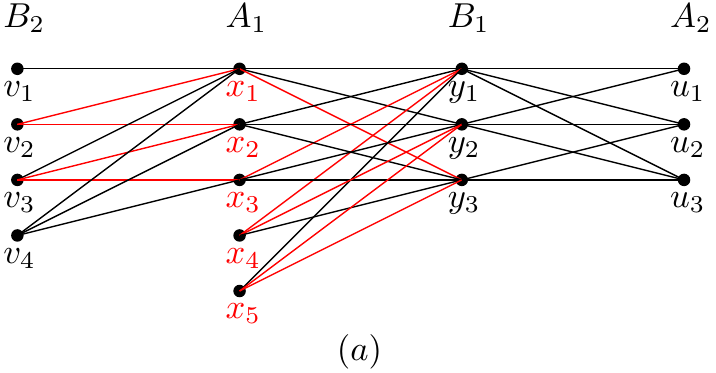}
		\includegraphics[width=80mm,scale=0.5]{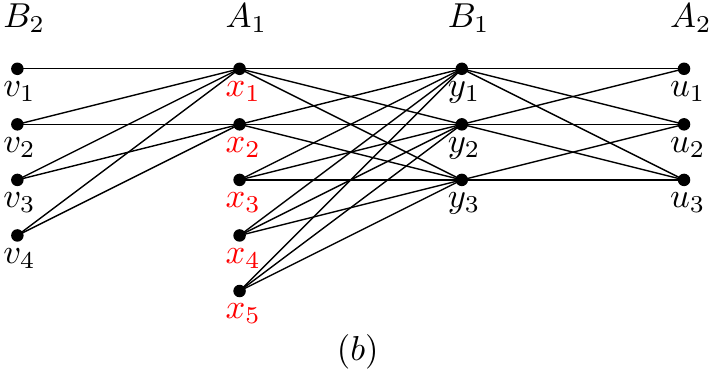}
		\caption{\small \sl  An example graph (a) Yes instance of Steiner cycle problem, (b) No instance of Steiner cycle problem.\label{fig:stcycle}}
	\end{center}
\end{figure}
{\bf Remark:} Proofs of the problems discussed in this section are constructive in nature and thus we obtain a polynomial-time algorithm for generalization of the  Hamiltonian cycle (path).
\section{Treewidth and Pathwidth of $P_5$-free chordal bipartite graphs}
In general, many graph-theoretic problems have linear-time algorithms in trees.   This motivates us to investigate the tree-like representation of graphs so that algorithmic techniques used in trees can be used to study the computational complexity of those problems in trees.  The study of treewidth investigates this line of study.  This problem was introduced by Robertson and Seymour \cite{Robertson}. This problem is NP-complete for bipartite graphs \cite{twbi} and planar graphs. Treewidth can be obtained for cographs in $\mathit{O(n)}$ time and for chordal bipartite graphs in $\mathit{O(n^6)}$ time \cite{kloks1995treewidth}. We have the following result in \cite{kloks1995treewidth} for chordal bipartite graphs. Given a chordal bipartite graph $G$, and a positive integer  $k$, if the treewidth of $G$ is at most $k$ then it returns yes with the underlying chordal embedding of $G$, otherwise, it says no. However, the exact value of treewidth is not explicitly mentioned due to the inherent complex structure of chordal bipartite graphs. Interestingly, for $P_5$-free chordal bipartite graphs, we present an exact value for treewidth as well as the underlying chordal embedding in linear time.
   In this section, we present a polynomial-time algorithm to find a tree decomposition of a $P_5$-free chordal bipartite graph, and hence the treewidth can be computed in polynomial time.  \\
\noindent\textbf{Definition : }\\
A tree decomposition  of a graph $G$ is a pair $(T,Z)$  where $T$ is a tree and $V(T)=Z=\{z_1,z_2,\ldots,z_k\}$,  $z_t$ is a vertex in $T$ such that $z_t \subseteq V(G)$ for each $t$, $1\leq t \leq k$.   Further, 

\begin{enumerate}[label=(\roman*)]
	\item $V(G)= \bigcup_{z_t\in V(T)}z_t$,
	\item For every edge $\{x,y\}\in E(G)$, there exists some $z_t\in V(T)$ such that $\{x,y\} \subseteq z_t$ and 
	\item For every vertex $u\in V(G)$ the set $\{z_t \in V(T)|u \in z_t\}$ induces a subtree in the tree $T$.
\end{enumerate}  
The width of a tree decomposition $(T,Z)$ is $\max_{z_t \in V(T)}~|~z_t|$.  The treewidth of a graph $G$, denoted by $tw(G)$, is the least possible width of a tree decomposition of $G$ minus one.  The vertex set $z_t$ is also referred to as $label(t)$.  In this section, we shall present an algorithm that outputs a tree decomposition of $G$.
\begin{algorithm}[!ht]
	\label{alg1}
	\BlankLine
	\SetAlgoLined
	\KwIn{A $P_5$-free chordal bipartite graph $G$}
	\KwOut{A tree decomposition of $G$}
	\BlankLine
	Let $A_1=(x_1,\ldots,x_i)$, $B_1(y_1,\ldots,y_j)$, $A_2(u_1,\ldots,u_p)$, and $B_2(v_1,\ldots,v_q)$ be the vertices of $G$ and $t=2$\\
	Let $\{z_1,\ldots,z_k\}$ be the vertices of $V(T)$\\
	\eIf{$(i+d(u_p))\leq(j+d(v_q))$}{Create a vertex $z_1$ in $T$ whose label, label $(z_1)=\{x_1,x_2,\ldots,x_i,y_1,y_2,\ldots,y_{d(u_p)}\}$\\ 
		\For{ $k=d(u_p)+1$ to $j$}{Create a label, $(z_t) =\{x_1,x_2,\ldots,x_i,y_k\}$ and add an edge $\{z_{t-1},z_t\}$ to $E(T)$\\$t++$ }
		\For{$k=q$ to 1}{Create a label, $(z_t)=\{v_k,N_G(v_k)\}$ and add an edge $\{z_t,z_{t-1}\}$ to $E(T)$\\$t++$}
		\For{$k=p$ to 1}{\eIf{ $k=p$}{Create a label, $(z_t)= \{u_k,N(u_k)\}$ and add an edge $\{z_t,z_{1}\}$ to $E(T)$\\$t++$}{Create a label $(z_t)= \{u_K,N(u_k)\}$ and add an edge $\{z_t,z_{t-1}\}$ to $E(T)$\\$t++$}
	}}{Create a vertex $z_1$ in $T$ whose label, label $(z_1)=\{y_1,y_2,\ldots,y_j,x_1,x_2,\ldots,x_{d(v_q)}\}$\\ 
		\For{ $k=d(v_q)+1$ to $i$}{Create a label, $(z_t) =\{y_1,y_2,\ldots,y_j,x_k\}$ and add an edge $\{z_{t-1},z_t\}$ to $E(T)$\\$t++$ }
		\For{$k=p$ to 1}{Create a label, $(z_t)=\{u_k,N_G(u_k)\}$ and add an edge $\{z_t,z_{t-1}\}$ to $E(T)$\\$t++$}
		\For{$k=q$ to 1}{\eIf{ $k=q$}{Create a label, $(z_t)= \{v_k,N(v_k)\}$ and add an edge $\{z_t,z_{1}\}$ to $E(T)$\\$t++$}{Create a label $(z_t)= \{v_k,N(v_k)\}$ and add an edge $\{z_t,z_{t-1}\}$ to $E(T)$\\$t++$}
	}}
	\caption{Tree Decomposition for $P_5$-free chordal bipartite graphs}
\end{algorithm}\\\\ 

\noindent \textbf{Correctness of Algorithm \ref{alg1}}
\begin{lemma}
	The graph $T$ obtained by Algorithm \ref{alg1} is a tree and it satisfies all the three properties of tree decomposition mentioned in the definition.
\end{lemma}
\begin{proof}
\end{proof}
\begin{enumerate}[label=(\roman*)]
	\item $T$ is a Tree\\
	Let $V(T)=\{z_1,z_2,\ldots,z_k\}$, and $E(T)=\{\{z_r,z_t\}|z_t$ is adjacent to $z_r$ by steps 6, 10, 15, 18 of Algorithm 1$\}$. From the construction, it is clear that the invariant maintained by Algorithm \ref{alg1} is connectedness.  We observe that at each step, $(z_t)$, $t>1$ is created and an edge $\{z_t,z_{t-1}\}$ is added to $E(T)$ corresponding to sets $z_t$ and $z_{t-1}$ which shows that $T$ is acyclic. This implies that $T$ is connected and acyclic.  Hence $T$ is a tree.\\
	\item For every edge $\{x,y\}\in E(G)$, there exist some $z_t\in V(T)$ such that $\{x,y\} \in label(z_t).$\\
	By our construction, at each step a vertex $x$ and $N_G(x)$ is added to $label(z_t)$ for some $t$.  Hence there exist some $z_t\in V(T)$ such that $\{x,y\} \in label(z_t)$. This shows that $V(G)= \bigcup_{z_t \in V(T)} label(z_t)$.\\
	\item For every vertex $u\in V(G)$ the set $\{z_t\in V(T)|u\in label(z_t)\}$ induces subtree of the tree $T$.\\
	As per the construction, the vertices $\{x_1,\ldots,x_i,y_1,\ldots,y_{d(u_p)}\}$ are added to the label $z_1$.  It is clear from Steps (5-20), the vertices $\{y_{d(u_p)},\ldots,y_j \}$ of $B_1$ and  $A_2$ are included based on {\em NNO} and an edge$\{z_t,z_{t-1})\}$ is added to $E(T)$.  Similarly, from Steps (24-40), the vertices of $A_1$ and  $B_2$ are included.  Observe that for each $u\in V(G)$, the set $\{z_t|u\in label(z_t)\}$ induces a subtree in $T$.  
\end{enumerate}

\begin{theorem}\label{twt}
	The treewidth of a $P_5$-free chordal bipartite graph is $tw(G)\geq \min \{(i+d(u_p)), (j+d(v_q))\}-1$.
\end{theorem}
\begin{proof}
	We observe that in any tree decomposition there exists a label, $z_t$ whose cardinality is at least $i+d(u_p)$. On the contrary, suppose there exists a tree decomposition such that $|label(z_t)|\leq i+d(u_p)-1$.  Without loss of generality, we assume that $z_t < i+d(u_p)$ and label$(z_t)=\{x_1,\ldots,x_i,y_1,\ldots,y_{d(u_p)-1}\}$. By the property of {\em NNO}, the vertex $u_p$ is adjacent to $\{y_1,\ldots,y_{d(u_p)\}}$.   Clearly, either the set $\{z~|~y_{d(u_p)}\in $ label$(z)\}$ or $\{z~|~u_p\}\in $ label$(z)\}$ does not form a subtree in $T$.  Therefore, $y_{d(u_p)} \in $ label$(z_1)$.  Similarly, each $y_r\in\{y_1,\ldots,y_{d(u_p)-1}\}$ must be part of label$(z_1)$. \qed

\end{proof}
\begin{lemma}
\label{tw}
The tree $T$ obtained from Algorithm \ref{alg1} is a tree decomposition of $G$ such that the treewidth,  $tw(G)=\min \{(i+d(u_p)), (j+d(v_q))\}-1$.
\end{lemma}
\begin{proof}
	By our construction, we see that the cardinality of each label$(z_t)$ is at most $\min \{(i+d(u_p)), (j+d(v_q))\}-1$, in particular, $|label(z_1)|=\min \{(i+d(u_p)), (j+d(v_q))\}-1$. \qed
\end{proof}
\begin{corollary}
	Let $G$ be a $P_5$-free chordal bipartite graph.  Then, the pathwidth of $G$ is $min\{(i+d(u_p)), (j+d(v_q))\}-1$.
\end{corollary}	
\begin{proof}
	We know that for any graph $G$, the pathwidth of $G$ is at least the treewidth of $G$.   By Theorem \ref{twt}, the treewidth of a $P_5$-free chordal bipartite graph is $\min \{(i+d(u_p)), (j+d(v_q))\}-1$. We observe that the tree decomposition of $G$ output by  Algorithm \ref{alg1} is a also a path decomposition of $G$.  Hence the pathwidth of $G$ is same as the treewidth of $G$ which is $min\{(i+d(u_p)), (j+d(v_q))\}-1$. \qed
\end{proof}
\noindent\textbf{Trace of the Treewidth (Pathwidth) Algorithm (Algorithm \ref{alg1}):} \\Consider the illustration given in Figure \ref{fig:Lp1}.   We shall trace the treewidth algorithm for Figure. \ref{fig:Lp1}.  By Step 3, we create a label $z_1=\{x_1,x_2,x_3,x_4,x_5,y_1,y_2,y_3\}$.  The trace of other steps are included in Figure \ref{fig:td1} and Figure \ref{fig:td2}.
\begin{figure}
	\begin{center}
		\includegraphics[width=35mm,scale=0.4]{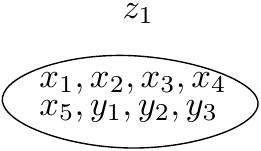}
		\caption{\small \sl   Step 3: Initialization of $label(z_1)$.\label{fig:td}}
	\end{center}
\end{figure}

\begin{figure}
	\begin{center}
		\includegraphics[width=167mm,scale=0.5]{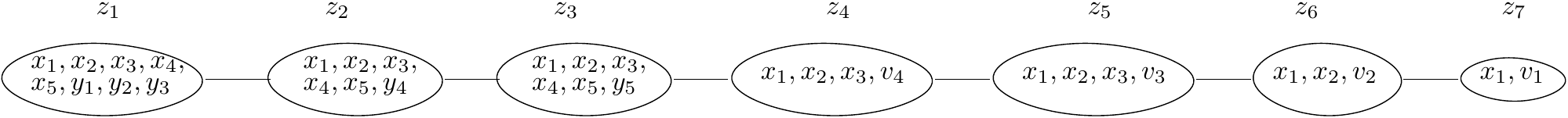}
		\caption{\small \sl   Steps (5-12): Vertices of $B_1$ and $B_2$ are included in $T$.\label{fig:td1}}
	\end{center}
\end{figure}

\begin{figure}[!h]
	\begin{center}
		\includegraphics[width=167mm,scale=0.5]{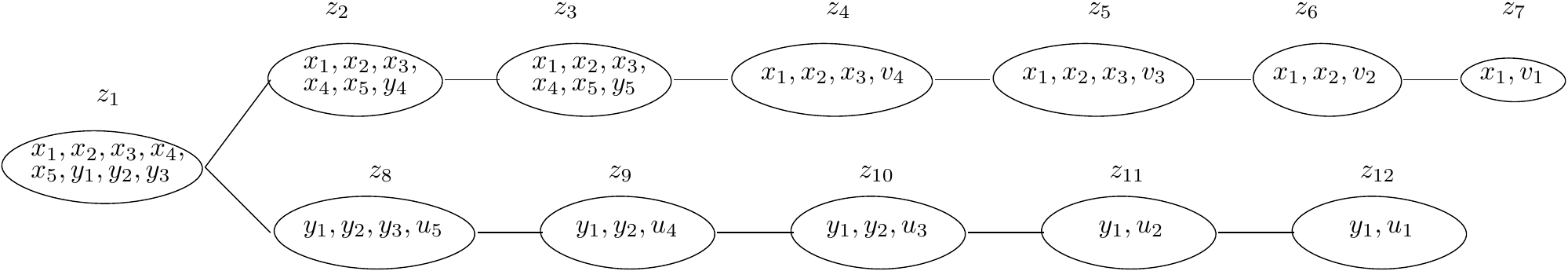}
		\caption{\small \sl   Steps (13-20): Vertices of $A_2$ are included in $T$.\label{fig:td2}} 
	\end{center}
\end{figure}
\section{Minimum Fill-in in $P_5$-free chordal bipartite graphs}
Having studied treewidth, we focus on problems that are related to treewidth.  We know that a graph $G$ has treewidth at most $k$ if and only if $G$ has a chordal completion whose maximum clique size is at most $k + 1$. This motivates us to study the complexity of the chordal completion problem in $P_5$-free chordal bipartite graphs. For a graph $G$, the minimum fill-in (chordal completion) \cite{kloks1993minimum} problem is the problem of finding a chordal embedding of $G$ which is a graph $G^*$ obtained by an augmenting minimum number of edges to $G$ so that $G^*$ is chordal.  The minimum fill-in problem is NP-complete on general graphs \cite{Minfill_g}. In \cite{kloks1993minimum} an $\mathit{O(n^5)}$ time algorithm is presented for chordal bipartite graphs. Subsequently, in \cite{Minfill_04} it is improved to $\mathit{O(n^4)}$. Interestingly, a chordal embedding of a $P_5$-free chordal bipartite graph is a split graph. In this section, we present a linear-time algorithm to find the minimum fill-in for $P_5$-free chordal bipartite graphs.\\\\

\begin{algorithm}[H]
	\label{alg2-mf} 
	\BlankLine
	\SetAlgoLined
	\KwIn{A $P_5$-free chordal bipartite graph $G$}
	\KwOut{A minimum fill-in of $G$}
	\BlankLine
	Let $A_1=(x_1,\ldots,x_i)$, $B_1=(y_1,\ldots,y_j)$, $A_2=(u_1,\ldots,u_p)$, and $B_2=(v_1,\ldots,v_q)$ be the vertices of $G$ \\
	Let $k=(i-d(v_q))$ and $l=(j-d(u_p))$\\
	\For{$r=d(u_p)$ to 1}{\For{$z=r-1$ to 1}{Add an edge $\{y_r,y_z\}$ to $E(G)$}}
	\For{$r=d(v_q)$ to 1}{\For{$z=r-1$ to 1}{Add an edge $\{x_r,x_z\}$ to $E(G)$}}
	\eIf{$(k*(i-k))\leq(l*(j-l))$ and $(k\leq l)$}{\For{$r=i$ to $(i-k)$}{\For{$z=r-1$ to 1}{Add an edge $\{x_r,x_z\}$ to $E(G)$}} 
	}{\For{$r=j$ to $(j-l)$}{\For{$z=r-1$ to 1}{Add an edge $\{y_r,y_z\}$ to $E(G)$}}}
	\caption{Minimum Fill-in of $P_5$-free chordal bipartite graphs}
\end{algorithm}
\begin{lemma}
	\label{lemmaminfill}
	Let $G$ be a $P_5$-free chordal bipartite graph. The number of edges augmented to find a chordal embedding of $G$ is min\{$(\frac{i(i-1)}{2} +  \frac{d(u_p)(d(u_p)-1)}{2}), (\frac{j(j-1)}{2} + \frac{d(v_q)(d(v_q)-1)}{2})\}$
\end{lemma}
\begin{proof}
Let $G^*$ be a chordal embedding of $G$ with $V(G^*)=V(G)$ and $E(G^*)=E(G^*)\cup S$, where $S$ is the set of augmented edges. We know that $A_1 \cup B_1$ is a maximal biclique in $G$. So, one of the partitions $A_1~ (B_1)$ of $G$  must be a clique $K$ in $G^*$. We consider the followings cases to find the edges added from $A_2$ $(B_2)$ to $K$.\\
	Case 1:  $A_1$ is a clique. Then, in any $G^*$, $\{y_1,y_2,\ldots,y_{d(u_p)}\}$ must be a clique.
	On the contrary, there exist a pair of vertices, say $y_r$, $y_s$ $\in$$B_1$, $1\leq r \leq d(u_p)$, $1\leq s \leq d(u_p)$ such that $\{y_r,y_s\}\notin E(G^*)$.  Due to {\em NNO} property in $G^*$, we get an induced cycle $C=(y_r,u_p,y_s,x_1,y_r)$ of length four. This contradicts the fact that $G^*$ is chordal. Therefore, $\{y_1,y_2,\ldots,y_{d(u_p)}\}$ induces a clique and the number of edges augmented is $\frac{d(u_p)(d(u_p)-1)}{2}$. \\
	Case 2: $B_1$ is a clique. Then, $\{x_1,x_2,\ldots,x_{d(v_p)}\}$ must be a clique. 
Suppose not, then there exist a pair of vertices, say $x_r$, $x_s$ $\in$$A_1$, $1\leq r \leq d(v_q)$, $1\leq s \leq d(v_q)$ that are not adjacent in $G^*$. Similar to Case 1, we get an induced cycle $C=(x_r,v_q,x_s,y_1,x_r)$ of length four, which is a contradiction. Therefore, we augment $\frac{d(v_q)(d(v_q)-1)}{2})\}$ edges to $A_1$.\\
	From the above case analysis, Lemma \ref{lemmaminfill} follows. \qed
	\end{proof}
\noindent\textbf{Correctness of Algorithm 2}
\begin{lemma}
	For a $P_5$-free chordal bipartite graph $G$, Algorithm \ref{alg2-mf} outputs a  minimum fill-in of $G$. Besides, Algorithm \ref{alg2-mf} outputs a chordal embedding of $G$.
	 \end{lemma}
\begin{proof}
	By our construction, the graph induced on $A_1$ ($B_1$) is a complete graph. This implies that the graph induced on $A_1 \cup B_1$ is a split graph, which is a chordal graph.  We know that the pendant vertices of $A_2$ cannot be a part of a cycle.  It is clear from the Steps (3-7) and the fact that $A_2$ follows {\em NNO} in $G$, $\{u_r\}\cup N_G(u_r)$, $1\leq r\leq p$, forms a clique.  Similarly, for $B_2$ as well.  We observe that there does not exist an induced cycle of length at least four in $G^*$. Therefore, the graph constructed by Algorithm \ref{alg2-mf} is a chordal graph. From Steps (3-7) and Steps (14-18), it is clear that the number of edges added is the number mentioned in Lemma \ref{lemmaminfill}, which is min\{$(\frac{i(i-1)}{2} +  \frac{d(u_p)(d(u_p)-1)}{2}), (\frac{j(j-1)}{2} + \frac{d(v_q)(d(v_q)-1)}{2}) \}$.  \qed
	\end{proof}
\noindent\textbf{Trace of Algorithm \ref{alg2-mf}:} 
Consider the illustration given in Figure \ref{fig:mf}. We trace Algorithm \ref{alg2-mf} for Figure \ref{fig:mf}. The degree of $u_p$, $d(u_p)=3$, Steps (3-5) of Algorithm \ref{alg2-mf} make $\{y_1,y_2,y_3\}$ a clique.  Similarly, $d(v_q)=3$, $\{x_1,x_2,x_3\}$ forms a clique.  Here, $k=l=3$, $(k*(i-k))\leq(l*(j-l))$ = 6$\leq$6.  Steps (14-18), initially chooses $u_5$ and add edges to $\{u_4,u_3,u_2,u_1\}$.  At Iteration 2, $u_4$ is connected to $\{u_3,u_2,u_1\}$.  The resultant graph $H$ is a chordal graph.
\begin{figure} 
	\begin{center}
		\includegraphics[width=100mm,scale=0.5]{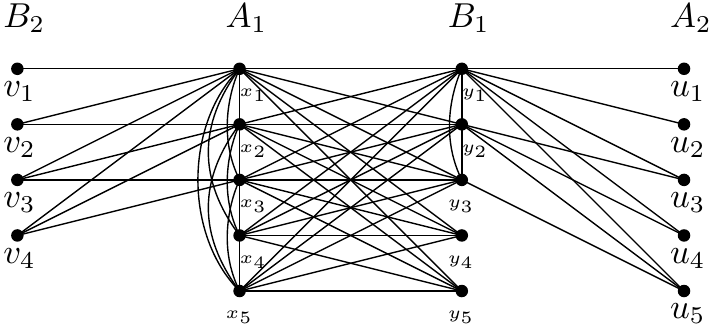}
		\caption{\small \sl   Minimum Fill-in of G.\label{fig:mf}} 
	\end{center}
\end{figure}\\

\begin{theorem}
	\label{split}
			The chordal embedding of a $P_5$-free chordal bipartite graph is a split graph.
	\end{theorem}
\begin{proof}
	By construction, it is clear that $A_1\cup B_1$ is a maximal biclique, and  $A_2$, $B_2$ are independent sets. Let $G^*$ be a chordal embedding of $G$.  Algorithm \ref{alg2-mf} makes $A_1 \cup \{y_1,y_2,\ldots,y_{d(u_p)}\}$ as a clique. The set $A_2\cup B_2  \cup \{y_{d(u_p)}, y_{d(u_p)+1},\ldots, y_j\}$ forms an independent set. We see that the graph $G^*$ is partitioned into an independent and a clique. Therefore, $G^*$ is a split graph.
	\end{proof}
\section{Results on Complement graphs of $P_5$-free chordal bipartite graphs} 
The generalization of the minimum fill-in problem is cluster editing problem  \cite{fomin2020subgraph} where we modify the input graph by adding/removing at most $k$ edges. Note that finding the complement of a graph is an example of a cluster editing problem.  Interestingly, the complement of $P_5$-free chordal bipartite graph is chordal. We now present some results on complement graphs of $P_5$-free chordal bipartite graphs. \\
We use the following notation, definition and results to prove our results.
Let $\overline{G}$ denote the complement of the graph $G$, where $V (\overline{G}) = V (G)$, and $E(\overline{G}) = \{\{u, v\} | \{u, v\} \notin E(G)\}$. Since $V(G)=V(\overline{G})$, the notation we follow for $G$ is applicable for $\overline{G}$ also.


\begin{definition}
	A vertex $x$ of $G$ is called simplicial if $N(x)$ induces a complete subgraph of $G$.
\end{definition}

\begin{definition}
Let $G = (V, E)$ be an undirected graph and let $\sigma = [v_1, v_2, \ldots, v_n]$ be an ordering of the vertices. We say that $\sigma$ is a perfect vertex elimination scheme (or perfect elimination ordering) if each $v_i$ is a simplicial vertex of the induced subgraph $G_{\{v_i, v_{i+1}, \ldots, v_n\}}$.
\end{definition}
\begin{theorem}
	\label{peo}
	(Dirac 1961, Fulkerson and Gross, 1965, Rose 1970). A graph $G$ is chordal if and only if $G$ has a PEO.
\end{theorem}	
	\begin{theorem}
		\label{chordal}
		If $G$ is a $P_5$-free chordal bipartite graph, then $\overline{G}$ is a chordal graph.
	\end{theorem}
	\begin{proof}
		To show that $\overline{G}$ is chordal, we now exhibit a {\em Perfect Elimination Ordering (PEO)} in $\overline{G}$. Considering an ordering $ \sigma $ of $V(\overline{G})$, $\sigma=(z_1=y_1,y_2,\ldots y_{j},v_1,v_2,\ldots,v_q,u_p,u_{p-1},\ldots,u_1,x_i,$ $x_{i-1},\ldots,x_1=z_n)$.		
		We now show that $\sigma$ is a {\em PEO}. On the contrary, assume that $\sigma$ is not a {\em PEO}. Let $z_r$ be the first vertex in $ \sigma $ whose neighborhood is not a clique in the graph induced on $\{z_r, z_{r+1}, \ldots, z_n\}$.  \\	\\
		Case 1: $z_r \in A(\overline{G})$ \\ 
		We claim that $A(\overline{G})$ is a clique. Suppose not, then there exist $u,v$ such that $\{u,v\}\notin E(\overline{G})$. This shows that in $A(G)$, $\{u,v\}\in E(G)$, a contradiction to the independent set property of $A(G)$. Note that $N_{\overline{G}}(z_r)=\{z_r, z_{r+1}, \ldots, z_n\}$ is a subset of $A(\overline{G})$.  Therefore, $z_r$ is a simplicial vertex. \\\\
		Case 2: $z_r \in B_1(\overline{G})$.\\
		Let $s=|N(z_r)\cap A_2(G)|$.   This implies that $w=p-s=|N_{\overline{G}}(z_r)\cap A_2(\overline{G})|$.  Since $A_2(G)$ satisfies {\em NNO}, $u_w$ to $u_p$  vertices of $A_2(\overline{G})$ are adjacent to the vertices of the set $\{y_{r+1},y_{r+2},\ldots, y_j\}$, and $A_2(\overline{G}) \cup B_2(\overline{G})$ is a clique. Let $H= A_2(\overline{G})\cup B_2(\overline{G}) \cup \{y_{r}, y_{r+1},\ldots,y_n\}$.
		We now claim that the graph induced on $H$ is a clique.  Suppose not, we obtain a contradiction to the independent set property of $G$ similar to Case 1. Note that $N_{\overline{G}}(z_r)=\{z_r, z_{r+1}, \ldots, z_n\}$ is a subset of $H$.
		Therefore,   $z_r$ is a simplicial vertex.  Similar argument is true for $z_r \in B_2(\overline{G})$.\\\\		
		From the above two cases, we conclude that $\sigma$ is a {\em PEO}.  Hence by Theorem \ref{peo},  $\overline{G}$ is a chordal graph.\\
		Suppose $A_2(G)=\emptyset$ or $B_2(G)=\emptyset$, then $A_1$ and $B_1$ are clique in 	
		$\overline{G}$. Therefore, $\overline{G}$ is chordal.  \qed
			\end{proof}
	 \begin{lemma}
	 	\label{vcin}
	 	 For a graph $G$, if the vertex connectivity is at least the independence number then $G$ is Hamiltonian \cite{chvatal1972note}.
	 	\end{lemma}
\begin{theorem}
	\label{cham}
	Let $G(A_1,B_1,A_2,B_2)$ be a $P_5$-free chordal bipartite graph such that $A_2$ and $B_2$ are non empty.  If $G$ is Hamiltonian, then $\overline{G}$ is Hamiltonian.
\end{theorem}
\begin{proof}
	Since $A_1$ $(B_1)$ is an independent set in $G$, it is clique in $\overline{G}$.  Similarly,  $A_2 \cup B_2$ is a clique in $\overline{G}$.  
	We observe that any independent set in $\overline{G}$ contains at most two vertices, in particular one vertex from $A(\overline{G})$ and other from $B(\overline{G})$. Therefore, the independence number is at most two. We observer that $\overline{G}$ is connected and now we show that there does not exist a cut vertex in $\overline{G}$. On the contrary, assume that there exists a cut vertex $z$. W.l.o.g. assume that $x$ and $y$ be the neighbors of $z$ in $\overline{G}$ .\\ \\
	Case 1: $x, y \in A(\overline{G})$. Since the graph induced on $A(\overline{G})$ is a clique, there exists a path $P(x,y)$ in $\overline{G}-z$. This shows that $\overline{G}-z$ is connected. Similar argument is true for $x,y \in B(\overline{G})$, and $x, y \in A_2 \cup B_2$ in $\overline{G}$.\\\\
	Case 2: $x\in A(\overline{G})$, and $y\in B(\overline{G})$.  Since $A_1$ $(B_1)$, and the set $A_2 \cup B_2$ induces a clique, we obtain a path $P(x,y)=(x,u_1,y_j,y)$, where $u_1\in A_2$, and  $y_j\in B_2$ in $G-z$. This shows that $z$ is not a cut vertex. Therefore, the vertex connectivity of $\overline{G}$ is at least two. From Lemma \ref{vcin}, we conclude that $\overline{G}$ is Hamiltonian.
\end{proof}
\textbf{Time Complexity Analysis :}\\
The {\em  Nested Neighbourhood Ordering (NNO)} of $P_5$-free chordal bipartite graphs can be obtained in linear time by using perfect edge elimination ordering proposed by Fulkerson and Gross. Further, the Hamiltonian cycle (path) can be obtained in linear time. Since bipancyclic, and homogeneously traceable involves order $n$ call to Hamiltonian cycle (path) algorithm, these two incurs $\mathit{O(n^2)}$ time.  All other algorithmic results run in linear time.\\\\
\textbf{Conclusions and Further Research}\\
 In this paper, we have presented structural results on $P_5$-free chordal bipartite graphs.  Subsequently, using these results, we have presented polynomial-time algorithms for the Hamiltonian cycle (path), also its variants and generalizations.   We also presented polynomial-time algorithms for combinatorial problems such as treewidth (pathwidth), and minimum fill-in.  We have also shown that Chv{\'a}tal's necessary condition is sufficient for this graph class.   Further, we have presented some results on complement graphs of $P_5$-free chordal bipartite graphs.  These results exploit the nested neighborhood ordering of $P_5$-free chordal bipartite graphs which is an important contribution of this paper.  A natural direction for further research is to study $P_6$-free chordal bipartite graphs and $P_7$-free chordal bipartite graphs as the complexity of most of these problems are open in these graph classes.

\end{document}